\documentclass{llncs}

\usepackage[utf8]{inputenc} 

\newcommand{\sfb}{\sf\bf}

\usepackage{natbib}
\usepackage{graphicx}

\usepackage{booktabs} 
\usepackage{array} 
\usepackage{hyperref} 
\usepackage{xspace} 
\usepackage{amsmath,amssymb} 
\usepackage{bm}
\usepackage{verbatim}
\usepackage{longtable}
\usepackage[vlined]{algorithm2e}

\usepackage{xifthen}
\usepackage{stmaryrd}

\usepackage{xcolor}

\usepackage[a4paper, left=1in, right=1in, top=1.2in]{geometry}

\hypersetup{
    bookmarks=false,         
    unicode=false,          
    pdftoolbar=true,        
    pdfmenubar=true,        
    pdffitwindow=true,     
    pdfstartview={FitH},    
    colorlinks=true,       
    linkcolor=red!70!black,          
    citecolor=red!70!black,        
    filecolor=magenta,      
    urlcolor=red!70!black           
}

\makeatletter
\newsavebox{\@alignepsbox}
\newlength{\@aligneps}
\newcommand{\includegraphicstop}[2][]{%
\sbox{\@alignepsbox}{\includegraphics[#1]{#2}}%
\setlength{\@aligneps}{-\ht\@alignepsbox}%
\addtolength{\@aligneps}{2ex}%
\raisebox{\@aligneps}{\usebox{\@alignepsbox}}}
\makeatother

\renewenvironment{proof}[1][]{\noindent \em Proof\ifthenelse{\equal{#1}{}}{}{ (#1)}:~}{}

\newcommand{\network}{\mathcal{N}}
\newcommand{\val}{\bar S} 
\newcommand{\dep}{\operatorname{dep}}
\newcommand{\energy}[1]{\operatorname{e}_{#1}}

\newcommand{\partfun}[1]{Z_{#1}}
\newcommand{\separator}[2]{\operatorname{sep}(#1,#2)}
\newcommand{\difference}[2]{\operatorname{diff}(#1 \rightarrow #2)}
\newcommand{\real}{\mathbb{R}}

\newcommand{\Message}[2]{m_{#1\rightarrow #2}}

\newcommand{\partseqs}{\mathcal{P\!S}}
\newcommand{\B}{\mathcal{B}}
\newcommand{\F}{\mathcal{F}}
\newcommand{\I}{\mathcal{I}}
\newcommand{\R}{\mathcal{R}}
\renewcommand{\S}{\mathcal{S}}
\newcommand{\X}{\mathcal{X}}
\newcommand{\Y}{\mathcal{Y}}

\newcommand{\width}{w}

\newcommand{\sample}{\texttt{Sample}}

\newcommand{\edgesToR}{E^r_T}

\newcommand{\Ebp}[2]{E^{\textrm{bp}}_{#1}(#2)}

\newcommand{\Def}[1]{{\bfseries #1}}

\newcommand{\TargetE}{E^{\star}}

\newcommand{\parHead}[1]{\Final{\paragraph{#1}}}

\newcommand{\Final}[1]{#1}
\renewcommand{\Final}[1]{}

\newcommand{\Design}[1]{{\sf Designs}^{\star}(#1)}
\newcommand{\NumDesign}{\ensuremath{\#}{\sf Designs}\xspace}
\newcommand{\IS}[1]{{\sf IndSets}(#1)}
\newcommand{\Nuc}[1]{{\sf #1}}
\newcommand{\Ab}{\Nuc{A}}
\newcommand{\Cb}{\Nuc{C}}
\newcommand{\Gb}{\Nuc{G}}
\newcommand{\Ub}{\Nuc{U}}

\newcommand{\GCb}{\Gb\Cb}

\newcommand{\Software}[1]{{\ttfamily #1}}

\newcommand{\ourprog}{\Software{RNARedPrint}}

\newcommand{\evalfor}[2]{#1\llbracket{}#2\rrbracket{}}
\newcommand{\substitute}[2]{#1\!\oplus\!#2}

\renewcommand{\gets}{:=}

\title
{Fixed-Parameter Tractable Sampling for RNA Design with Multiple Target Structures}
\author
{Stefan Hammer\,$^{\text{\sfb 1,2,3}}$, Yann Ponty\,$^{\text{\sfb 4,5,}\star}$, Wei Wang\,$^{\text{\sfb 4}}$ and Sebastian Will\,$^{\text{\sfb 2}}$}

\institute{$^{\text{\sf 1}}$University Leipzig, Department of Computer Science and Interdisciplinary Center for Bioinformatics, 04107 Leipzig, Germany;
$^{\text{\sf 2}}$University of Vienna, Faculty of Chemistry, Department of Theoretical Chemistry, 1090 Vienna, Austria;
$^{\text{\sf 3}}$University of Vienna, Faculty of Computer Science, Research Group Bioinformatics and
Computational Biology, 1090 Vienna, Austria;
$^{\text{\sf 4}}$CNRS UMR 7161 LIX, Ecole Polytechnique, Bat. Turing, 91120 Palaiseau, France;
and $^{\text{\sf 5}}$AMIBio team, Inria Saclay, Bat Alan Turing, 91120 Palaiseau, France}

\begin{document}

\maketitle

\begin{abstract}
\textbf{Motivation:} The design of multi-stable RNA molecules has important applications in biology, medicine, and biotechnology. Synthetic design approaches profit strongly from effective in-silico methods, which can tremendously impact their cost and feasibility. \\
  \textbf{Results:} We revisit a central ingredient of most in-silico
  design methods: the sampling of sequences for the design of
  multi-target structures, possibly including pseudoknots. For this
  task, we present the efficient, tree decomposition-based algorithm
  \ourprog{}. Our fixed parameter tractable approach is underpinned by
  establishing the $\#${\sf P}-hardness of uniform sampling. Modeling
  the problem as a constraint network, \ourprog{} supports generic
  Boltzmann-weighted sampling for arbitrary additive RNA energy
  models; this enables the generation of RNA sequences meeting
  specific goals like expected free energies or \GCb-content. Finally,
  we empirically study general properties of the approach and generate
  biologically relevant multi-target Boltzmann-weighted designs for a
  common design benchmark. Generating seed sequences with \ourprog{}, we demonstrate significant improvements over the previously best multi-target sampling strategy (uniform sampling).\\
  \textbf{Availability:} Our software is freely available at: \url{https://github.com/yannponty/RNARedPrint}\\
  \textbf{Contact:} \href{yann.ponty@lix.polytechnique.fr}{yann.ponty@lix.polytechnique.fr}\\
\end{abstract}

\section{Introduction}
\parHead{Design, applications and motivation for multiple design.}Synthetic biology endeavors the engineering of artificial biological
systems, promising broad applications in biology, biotechnology and
medicine. Centrally, this requires the design of biological
macromolecules with highly specific properties and programmable functions.
In particular, RNAs present themselves as well-suited tools for
rational design targeting specific functions~\citep{Kushwaha2016}. RNA function is tightly
coupled to the formation of secondary structure, as well as changes in
base pairing propensities and the accessibility of regions, e.g. by
burying or exposing interaction sites~\citep{Rodrigo2014}. At the same time, the
thermodynamics of RNA secondary structure is well understood and its prediction is
computationally tractable~\citep{McCaskill1990}. Thus,  structure can serve as effective
proxy within rational design approaches, ultimately targeting catalytic~\citep{Zhang2013} or regulatory~\citep{Rodrigo2014} functions.

\parHead{Motivating multiple RNA design.} The function of many RNAs
depends on their selective folding into one or several alternative
conformations. Classic examples include riboswitches, which
notoriously adopt different stable structures upon binding a specific
ligand. Riboswitches have been a popular application of rational
design~\citep{Wachsmuth2013,Domin2017}, partly motivated by their
capacity to act as biosensors~\citep{Findeiss2017}. At the
co-transcriptional level, certain RNA families feature alternative,
evolutionarily conserved, transient structures~\citep{Zhu2013}, which
facilitate the ultimate adoption of functional structures at full
elongation.  More generally, simultaneous compatibility to multiple
structures is a relevant design objective for engineering kinetically
controlled RNAs, finally targeting prescribed folding pathways. Thus,
modern applications of RNA design often target multiple structures,
additionally aiming at other features, such as specific
\GCb-content~\citep{Reinharz2013} or the presence/absence of
functionally relevant motifs (either anywhere or at specific
positions)~\citep{Zhou2013}; these objectives motivate flexible
computational design methods.

\parHead{On the importance of sampling for design.}
Many computational methods for RNA design rely on similar overall
strategies: initially generating one or several \Def{seed} sequences
and optimizing them subsequently. In many cases, the seed quality was
found to be critical for the empirical performance of RNA design
methods~\citep{Levin2012}. For instance, random seed generation
improves the prospect of subsequent optimizations, helping to overcome
local optima of the objective function, and increases the diversity
across designs~\citep{Reinharz2013}.  For single-target approaches,
\Software{INFO-RNA}~\citep{Busch2006} made significant improvements
mainly by starting its local search from the minimum energy sequence
for the target structure instead of (uniform) random sequences for the
early \Software{RNAinverse} algorithm~\citep{Hofacker1994}. This
strategy was later shown to result in unrealistically high
\GCb-contents in designed sequences. To address this issue,
\Software{IncaRNAtion}~\citep{Reinharz2013} controls the \GCb-content
through an adaptive sampling strategy.

\parHead{Specificities and similarities of multi-target design.}
Specifically, for multi-target design, virtually all available methods~\citep{Lyngsoe2012,HoenerzuSiederdissen2013,Taneda2015,Hammer2017} follow the same overall generation/optimiza\-tion scheme.
Facing the complex sequence constraints induced by multiple targets, early methods such as \Software{Frnakenstein}~\citep{Lyngsoe2012} and \Software{Modena}~\citep{Taneda2015} did not attempt to solve sequence generation systematically, but rely on \emph{ad-hoc} sampling strategies.
Recently, the \Software{RNAdesign}
approach~\citep{HoenerzuSiederdissen2013}, coupled with powerful local
search within \Software{RNAblueprint}~\citep{Hammer2017}, solved the
problem of sampling seeds from the uniform distribution for multiple
targets. These methods adopt a graph coloring perspective, initially
decomposing the graph hierarchically using various decomposition
algorithms, and \Def{precomputing} the number of valid sequences
within each subgraph. The decomposition is then reinterpreted as a
decision tree to perform a \Def{stochastic backtrack}, inspired by
Ding and Lawrence~\citep{Ding2003}. Uniform sampling is achieved by
choosing individual nucleotide assignments with probabilities derived
from the subsolution counts. The overall complexity of
\Software{RNAdesign} grows like $\Theta(4^{\gamma})$, where the
parameter $\gamma$ is bounded by the length of the designed RNA;
typically, the decomposition strategy achieves much lower $\gamma$.

\parHead{Motivation.} The exponential time and space requirements of the \Software{RNAdesign} method already raise the question of the \Def{complexity of (uniform) sampling for multi-target design}. Since stochastic backtrack can be performed in linear time per sample, the method is dominated by the precomputation step, which requires counting valid designs. Thus, we focus on the question: \emph{Is there a polynomial-time algorithm to count valid multi-target designs?} In Section~\ref{sec:counting}, we answer in the negative, showing that there exists no such algorithm unless ${\sf P}={\sf NP}$. Our result relies on a surprising bijection (up to a trivial symmetry) between valid sequences  and independent sets of a bipartite graph, being the object of recent breakthroughs in approximate counting complexity~\citep{Bulatov2013,Cai2016}.
The hardness of counting (and conjectured hardness of sampling) does not preclude, however, practically applicable algorithms for counting and sampling. In particular, we wish to extend the flexibility of multi-structure design, leading to the following questions: \emph{How to sample, generally, from a Boltzmann distribution based on expressive energy models? How to enforce additional constraints, such as the \GCb-content, complex sequence constraints, or the energy of individual structures?}

\begin{figure*}[t]
{\centering\includegraphics[width=.8\textwidth]{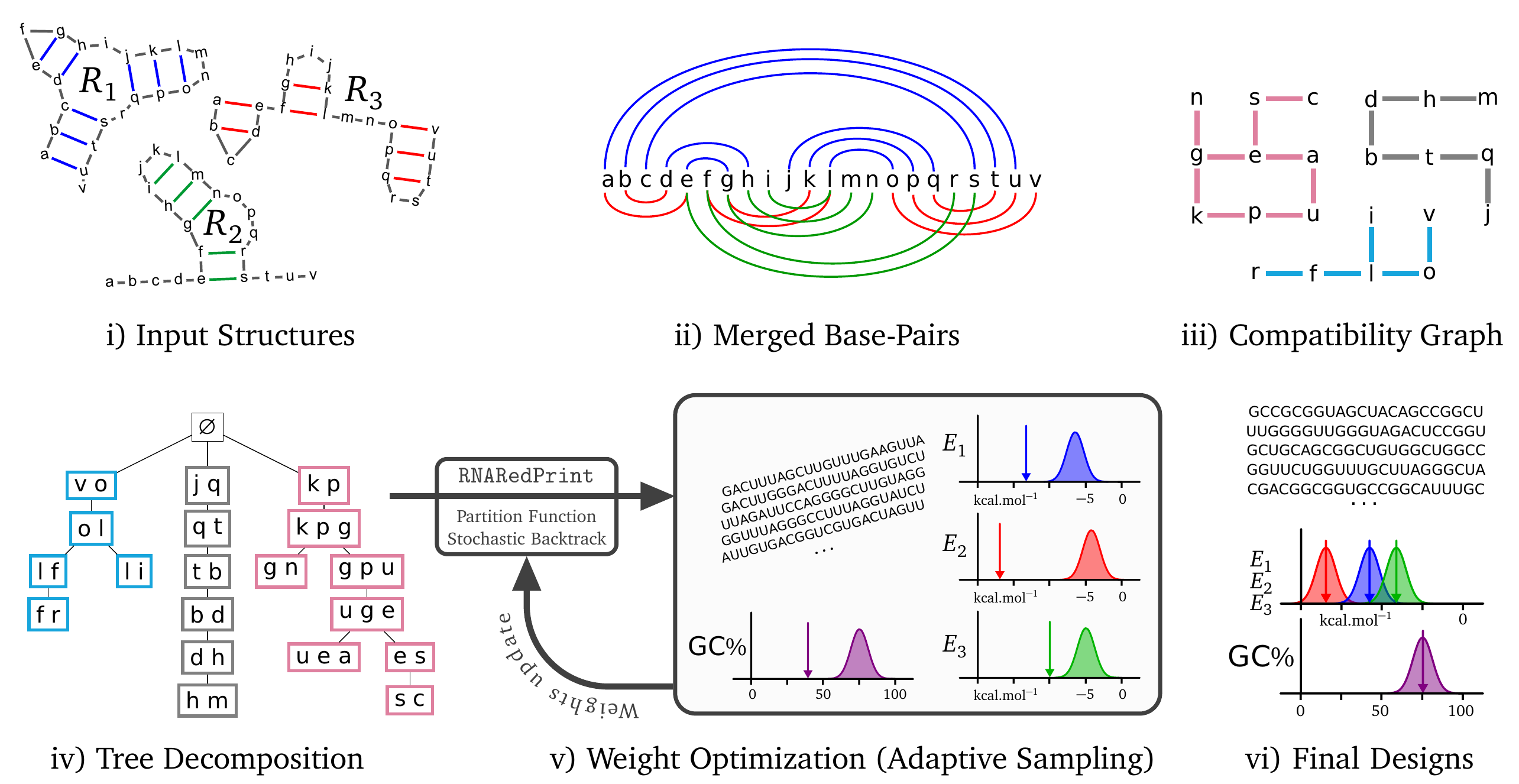}\\}
\caption{General outline of \ourprog{} for base pair-based energy models. From a set of target secondary structures (i), base-pairs are merged (ii) into a (base pair) dependency graph (iii) and transformed into a tree decomposition (iv). The tree is then used to compute the partition function, followed by a Boltzmann sampling of valid sequences (v). An adaptive scheme learns weights to achieve targeted energies and \GCb-content, leading to the production of suitable designs (vi).}
\label{fig:workflow}
\end{figure*}

To answer these questions, we introduce a generic framework (illustrated in~Fig.~\ref{fig:workflow}) enabling efficient Boltzmann-weighted sampling over RNA sequences with multiple target structures (Section~\ref{sec:FPT}). Guided by a \Def{tree decomposition} of the network, we devise dynamic programming to compute partition functions and sample sequences from the Boltzmann distribution%
. We show that these algorithms are \Def{fixed-parameter tractable} for the \Def{treewidth}%
%
; in practice, we limit this parameter by using state-of-the-art tree decomposition algorithms.
By evaluating (partial) sequences in a weighted constraint
network, we support arbitrary multi-ary constraints and thus
arbitrarily complex energy models,
notably subsuming all commonly
used RNA energy models%
.  Moreover, we describe an \Def{adaptive
  sampling} strategy to control the free energies of the individual
target structures and the \GCb-content%
. %
We observe that sampling based on less complex RNA energy models
(taking only the most important energy contributions into account)
still allows targeting realistic RNA energies in the well-accepted
Turner RNA energy model. The resulting combination of efficiency and
high accuracy finally enables generating biologically relevant
multi-target designs in our final application of our overall strategy
to a large set of multi-target RNA design instances from a
representative benchmark Section~\ref{sec:results}).

\section{Definitions and problem statement}
\label{sec:problem-statement}

An \Def{RNA sequence $S$} is a word over the \Def{nucleotides
  alphabet} $\Sigma=\{\Ab,\Cb,\Gb,\Ub\}$; let $\S_n$ denote the set of
sequences of length $n$. An \Def{RNA (secondary) structure $R$ of
  length $n$} is a set of \Def{base pairs} $(i,j)$, where
$1\leq i<j\leq n$, where for all different $(i,j), (i',j')\in R$:
$\{i,j\}\cap\{i',j'\}=\emptyset$ (``degree $\leq$ 1'').
%
%
\Def{Valid base pair} must pair bases from 
$\B:=\left\{\{\Ab,\Ub\},\{\Gb,\Cb\},\{\Gb,\Ub\}\right\}.$
Consequently, $S$ is \Def{valid} for $R$, iff $\{S_i,S_j\}\in \B$ for
all $(i,j)\in R$.

We consider a fixed set of target RNA structures
$\R:=\{R_1, \dots, R_k\}$ for sequences of length $n$. $\R$ induces a
\Def{base pair dependency graph} $G_{\R}$ with nodes $\{1,\dots,n\}$
and edges $\bigcup_{\ell\in[1,k]} R_\ell$, which describe the minimal
dependencies present in all relevant settings due to the requirement
of canonical base pairing.

One can interpret the valid sequences for $\R$ as colorings of
$G_{\R}$ in a slightly modified graph coloring variant, where the
colors (from $\Sigma$) assigned to adjacent vertices of $G_{\R}$ must
constitute valid base pairs.
We define the energy of a sequence $S$ based on the set of structures
$\R$ as $E_\R(S) \in \mathbb{R}\cup\{\infty\}$. In
our setting, the energy $E_\R(S)$ is additively composed of
the energies of the single RNA structures in an RNA energy model, as
well as sequence dependent features like \GCb-content.  Furthermore note
that $E_\R(S)$ is finite iff $S$ is valid for each structure
$R_1,\dots,R_k$.

At the core of this work, we study the computation of partition
functions over sequences.\smallskip\\
\textbf{Central problem (Partition function over sequences).}
  Given an energy function $E$ and a set $\R$ of structures of length $n$, compute the
  partition function
  \begin{equation}
    \label{eq:mainproblem}
    \partfun{E_\R} = \sum_{S\in\S_n} \exp(-\beta E_\R(S)),
  \end{equation}
  where $\beta$ denotes the inverse pseudo-temperature, and $\S_n$ the set of sequences of length $n$.

%
  As we elaborate in subsequent sections, our approach relies on
  breaking down the energy function $E_\R(S)$ into additive
  components, each depending on only few sequence positions. Given
  $\R$, we express $E_\R(S)$ as the sum of energy contributions $f(S)$
  over a set $\F$ (of functions $f:\S_n\to\real$),
  s.t.~$E_\R(S)=\sum_{f\in\F} f(S)$. This captures realistic RNA
  energy models---including nearest neighbor models,
  e.g.~\citet{Turner2009}, and even pseudoknot models,
  e.g. \citet{Andronescu2010}, while bounding the dependencies to
  sequence positions introduced by each single $f$. Formally, define
  the \Def{dependencies} $\dep(f)$ of $f$ as the minimum set of sequence
  positions $\I\subseteq\{1,\dots,n\}$, where $f(S)=f(S')$ for all
  sequences $S$ and $S'$ that agree at the positions in $\I$.

  Each set $\F$ of functions (on sequences of length $n$) induces a
  \Def{dependency graph} on sequence positions, namely the hypergraph
  $G_\F=(\{1,\dots,n\} ,\{\dep(f)\mid f\in \F\})$.  Our algorithms
  will critically rely on a \Def{tree decomposition} of the dependency
  graph, which we define below.
\begin{definition}[Tree decomposition and width]
  \label{def:treedecomp}
  Let $G=(X, E)$ be a hypergraph. A \Def{tree decomposition} of $G$ is
  a pair $(T,\chi)$, where $T$ is an unrooted tree/forest, and (for
  each $v\in T$) $\chi(v)\subseteq X$ is a set of vertices assigned to
  the node tree $v\in T$, such that
\begin{enumerate}
\item each $x\in X$ occurs in at least one $\chi(v)$;
\item for all $x\in X$, $\{ v \mid x \in \chi(v) \}$ induces a connected subtree of $T$;
\item for all $e\in E$, there is a node $v\in T$, such that $e\subseteq\chi(v)$.
\end{enumerate}
The \Def{width} of a tree decomposition $(T,\chi)$ is defined as
$\width(T,\chi) = \min_{u\in T} |\chi(u)| - 1 $. The \Def{treewidth}
of $G$ is the smallest width of any tree decomposition of $G$.
\end{definition}

\section{An FPT algorithm for the partition function and sampling of Boltzmann-weighted designs}
\label{sec:FPT}

For our algorithmic description, we translate the concepts of
Section~\ref{sec:problem-statement} to the formalism of constraint networks, here
specialized as RNA design network. This allows us to base our
algorithm on the cluster tree elimination (CTE) of~\citet{Dechter2013}.
In the RNA design network, (partially determined) RNA sequences
replace the more general concept of (partial) assignments in
constraint networks. Partially determined RNA sequences, for short
\Def{partial sequences}, are words $\val$ over the alphabet
$\Sigma\cup\{?\}$ equivalently representing the set ${\mathcal S}(\val)$
of RNA sequences, where for positions $1\leq i\leq n$,
$\val_i\in\Sigma$ implies $S_i=\val_i$ for all $S\in{\mathcal
  S}(\val)$. The positions $1\leq i\leq n$, where $\val_i\in\Sigma$,
are called \Def{determined} by $\val$ and form its \Def{domain}.
Since the functions $f\in\F$ of Section~\ref{sec:problem-statement}
depend on only the subset $\dep(f)$ of sequence positions, one can
evaluate them for partial sequences $\val$ that determine (at least)
the nucleotides at all positions in $\dep(f)$. Thus for functions $f$
and partial sequences $\val$ that determine $\dep(f)$, we write
$\evalfor{f}{\val}$ to \Def{evaluate $f$ for $\val$}; i.e.{}
$\evalfor{f}{\val} := f(S),$ for any sequence
$S\in{\mathcal S}(\val)$.

\begin{definition}
An \Def{RNA design network} (for sequences of length $n$) is a tuple $\network=(\X,\F)$, where\vspace{-6pt}
\begin{itemize}
\item $\X$ is the set of sequence positions $1,\dots,n$
\item $\F$ is a set of \Def{functions} $f:\S_n\to\real$
\end{itemize}
\end{definition}

The \Def{energy $\energy{\network}(S)$ of a sequence} $S$ in a network
$\network$ is defined as sum of the values of all functions in
$f\in\F$ evaluated for $S$, i.e.
$\energy{\network}(S) := \sum_{f\in\F} \evalfor{f}{S}.$

The network energy $\energy{\network}(S)$ corresponds to the energy in
Eq.~$(\ref{eq:mainproblem}),$ where this energy is modeled as sum of
the functions in $\F$. Consequently, $\partfun{\R}$ of
Eq.~$(\ref{eq:mainproblem})$ is modeled as network partition function
$\partfun{\network} := \sum_{S}\exp(-\beta\energy{\network}(S)) = \sum_{S}\prod_{f\in\F} \exp( -\beta\cdot
\evalfor{f}{S} ).$


\subsection{Partition function and Boltzmann sampling through stochastic backtrack}\label{sec:PF}
The minimum energy, counting, and partition function
over RNA design network can be computed by dynamic programming based
on a tree decomposition of the network's dependency graph
(i.e. cluster tree elimination).
We focus on the efficient computation of the partition
function. 

%

We require additional definitions: A \Def{cluster tree} for the
network $\network=(\X,\F)$ is a tuple $(T,\chi,\phi)$, where
$(T,\chi)$ is a tree decomposition of $G_\F$, and $\phi(v)$ represents
a set of functions $f$, each uniquely assigned to a node $v\in T$;
$\dep(f)\subseteq\chi(v)$ and $\phi(v)\cap \phi(v')=\varnothing$ for
all $v\neq v'$.  For two nodes $v$ and $u$ of the cluster tree, define
their \Def{separator} as $\separator{u}{v} := \chi(u)\cap\chi(v)$;
moreover, we define the \Def{difference positions} from $u$ to an
adjacent $v$ by $\difference{u}{v}:=\chi(v) - \separator{u}{v}$.

For a set $\Y$ of sequence positions, write $\partseqs(\Y)$ to
denote the set of all partial sequences that determine exactly the positions
in $\Y$; furthermore, given partial sequences $\val$ and
$\val'$, we define the \Def{combined partial sequence $\substitute{\val'}{\val''}$} such that
$$
(\substitute{\val}{\val'})_i :=
\begin{cases}
  \val'_i & \text{if } \val'_i\in \Sigma\\
  \val_i & \text{otherwise}
\end{cases}
$$

Finally we assume, w.l.o.g., that all position difference sets
$\difference{u}{v}$ are singleton: for any given
cluster tree, an equivalent (in term of treewidth) cluster tree can
always be obtained by inserting at most $\Theta(|\X|)$ additional
clusters.

Let us now consider the \Def{computation of the partition function}.
Given is the RNA design network $\network=(\X,\F)$ and its cluster
tree decomposition $(T,\chi,\phi)$.  W.l.o.g., we assume that $T$ is
connected and contains a dedicated node $r$, with
$\chi(r)=\varnothing$ and $\phi(r)=\varnothing$, added as a virtual
root $r$ connected to a node in each connected component of $T$.  Now,
we consider the set of directed edges $\edgesToR{}$ of $T$ oriented to
$r$; define $T_r(u)$ as the induced subtree of
$u$. Algorithm~\ref{alg:pf} computes the partition function by passing
messages along these directed edges $u\to v$ (i.e. always from some
child $u$ to its parent $v$). Each message is a function that depends
on the positions $\dep(m)\subseteq \X$ and yields a partition function
in $\real\cup\{\infty\}$. The message from $u$ to $v$ represents the
partition functions of the subtree of $u$ for all possible partial
sequences in $\partseqs(\separator{u}{v})$. Induction over $T$ lets us show
the correctness of the algorithm
(Supp. Mat.~\ref{appsec:correctness}).  After running
Alg.~\ref{alg:pf}, multiplying the 0-ary messages sent to the root $r$
yields the total partition function:
\begin{math}
  \partfun{\network} = \prod_{(u\to{}r)\in T} \evalfor{\Message{u}{r}}{\varnothing}.
\end{math}

\begin{algorithm}[t]
  \KwData{Cluster tree $(T,\chi,\phi)$} \KwResult{Messages
    $\Message{u}{v}$ for all $(u\to{}v)\in T$; i.e.~partition
    functions of the subtrees of all $v$ for all possible partial
    sequences determining exactly the positions $\separator{u}{v}$.}
 \For{$u\to{}v\in T$ in postorder}{
  \For{$\val\in\partseqs(\separator{u}{v})$}{
    $x\gets 0$\;
    \For{$\val'\in\partseqs(\difference{u}{v})$}{
     $p \gets$ product( $exp(-\beta \evalfor{f}{\substitute{\val}{\val'}})$ for $f\in \phi(u)$ )\\
     ${}\quad\qquad \cdot\ $product( $\evalfor{\Message{w}{u}}{\substitute{\val}{\val'}}$ for $(w\to{}u)\in T$ )\;
     $x \gets x + p$\;
   }
   $\evalfor{\Message{u}{v}}{\val} \gets x$\;
  }
  \Return {$m$}\;
  }
 \caption{FPT computation of the partition function using
   dynamic programming (CTE). }\label{alg:pf}
\end{algorithm}

\SetKwProg{Fn}{Function}{}{}
 \SetKwFunction{Sample}{$\sample$}
 \SetKwFunction{Random}{UnifRand}

The partition functions can then direct a \Def{stochastic backtrack} to achieve \Def{Boltzmann sampling of sequences}, such that one samples from the Boltzmann distribution of a given design network $\network$. The sampling algorithm assumes that the cluster tree was expanded and the messages $\Message{u}{v}$ for the edges in $\edgesToR{}$ are already generated by Algorithm \ref{alg:pf} for the expanded cluster tree.
Algorithm~\ref{alg:sampling} defines the recursive procedure
$\sample(u,\val)$, which returns---randomly drawn from the
Boltzmann distribution---a partial sequence that determines all
sequence positions in the subtree rooted at $u$.  Called on $r$ and
the empty partial sequence, which does not determine any positions,
the procedure samples a random sequence from the Boltzmann
distribution.


\subsection{Computational complexity of the multiple target sampling algorithm}\label{sec:complexity}


Note that in the following complexity analysis, we omit time and space for computing the
tree decomposition itself, since we observed that the computation time
of tree decomposition (\Software{GreedyFillIn}, implemented in
\Software{LibTW}~by \citet{Dijk2006}) for multi-target sampling is
negligible compared to Alg.~\ref{alg:pf}
(Supp. Mat.~\ref{appsec:treedecomp} and
\ref{appsec:dependency-cliques}).

We define the \emph{maximum separator size} $s$ as
$\max_{u,v\in V} | \separator{u}{v} |$ and denote the maximum size of
$\difference{u}{v}$ over $(u,v)\in\edgesToR{}$ as $D$.  In the absence
of specific optimizations, running Alg.~\ref{alg:pf} requires
$\mathcal{O}((|\F|+|V|)\cdot 4^{w+1})$ time and
$\mathcal{O}(|V|\cdot4^s)$ space
(Supp. Mat.~\ref{appsec:algcomplexity}); Alg.~\ref{alg:sampling} would
require $\mathcal{O}((|\F|+|V|)\cdot 4^D)$ per sample on arbitrary
tree decompositions
(Supp. Mat.~\ref{appsec:algcomplexity}). W.l.o.g. we assume that
$D=1$; note that tree decompositions can generally be transformed,
such that $\difference{u}{v}\leq 1$.
Moreover, the size of $\F$ is linearly bounded: for $k$
input structures for sequences of length $n$, the energy function is
expressed by $\mathcal{O}(n\,k)$ functions. Finally, the number of cluster
tree nodes is in $O(n)$, such that $|\F|+|V| \in \mathcal{O}(n\,k)$.

\begin{algorithm}
 \KwData{Node $u$, partial sequence $\val\in\partseqs(\separator{u}{v})$;\newline
 Cluster tree $(T,\chi,\phi)$ and partition functions $\Message{u'}{v'}[\val']$, $\forall (u'\to{}v')\in T$ and $\val'\in\partseqs(\separator{u'}{v'})$.}
 \KwResult{Boltzmann-distributed random partial sequence for the subtree rooted at $u$, specializing a partial sequence $\val$.}
 \Fn{\Sample$(u,\val;T,\chi,\phi,m)$}{
   $r \gets \Random(\evalfor{\Message{u}{v}}{\val})$\;
   \For{$\val'\in\partseqs(\difference{u}{v})$}{

     $p \gets$ product( $exp(-\beta \evalfor{f}{\substitute{\val}{\val'}})$ for $f\in \phi(u)$\\
     ${}\quad\qquad \cdot\ $product( $\evalfor{\Message{w}{u}}{\substitute{\val}{\val'}}$ for $(w\to{}u)\in T$ )\;
%
  	  $r \gets r - p$\;
  	  \If{$r<0$}{
  	    $\val_{\rm res} \gets \substitute{\val}{\val'}$\;
  	  \For{$(v\to{}w)\in T$}{
  	      $\val_{\rm res} \gets \substitute{\val_{\rm res}}{\Sample(v,\substitute{\val}{\val'};T,\chi,\phi,m)}$\;
     }
  	  \Return {$\val_{\rm res}$}
  	  \;}
 }
 }
 \caption{Stochastic backtrack algorithm for partial sequences in the Boltzmann distribution.}\label{alg:sampling}
\end{algorithm}

\begin{theorem}[Complexities]\label{th:complexities}
  For sequence length $n$, $k$ target structures, treewidth $w$ and
  a base pair dependency graph having $c$ connected components, $t$ sequences
  are generated from the Boltzmann distribution  in
  $O( n\, k \, 4^{w+1} + t\, n\, k )$ time.
\end{theorem}
As shown in Supp. Mat.~\ref{sec:improvedComplexity}, the complexity of the precomputation can be further improved to 
$\mathcal{O}(n\,k\,2^{w+1}\,2^{c})$, where $c$ is the maximum number of connected components represented in a node of the tree decomposition ($c\le w+1$).

\subsection{Sequence features, constraints, and energy
  models.}\label{sec:energy_models}

The functions in $\F$ allow expressing complex features of the
sequences alone, e.g. rewarding or penalizing specific sequence
motifs, as well as features depending on the target structures.
Furthermore, constraints, which enforce or forbid features, are
naturally expressed by assigning infinite penalties to invalid
(partial) sequences. Finally, but less obviously, the framework
captures various RNA energy models with bounded dependencies, which we
describe briefly.

In simple \Def{base pair-based energy models}, energy is defined as
the sum of base pair (pseudo-)energies. If base pair energies
$\Ebp{k}{i,j,x,y}$ (where $i$ and $j$ are sequence positions, $x$ and
$y$ are bases in $\Sigma$) are given for each target structure $\ell$,
s.t.  $ E(S;R_\ell) := \sum_{(i,j)\in R_\ell} \Ebp{k}{i,j,S_i,S_j}$,
we encode the network energy by the set of functions $f$ for each base
pair $(i,j)\in R_\ell$ of each input structure $R_\ell$ that evaluate
to $\ln(\pi_\ell) \Ebp{k}{i,j,\val(x_i),\val(x_j)}$ under partial sequence $\val$;
here, $\pi_\ell>0$ is a weight that controls the
influence of structure $\ell$ on the sampling (as elaborated later).



More complex \Def{loop-based}
 energy models ---e.g.~the Turner model,
which among others includes energy terms for special loops and dangling
ends---can also be encoded as straightforward extensions. An interesting
stripped-down variant of the nearest neighbor model is the
\Def{stacking energy model}. This model assigns non-zero energy
contributions only to stacks, i.e. interior loops with closing base
pair $(i,j)$ and inner base pair $(i+1,j-1)$.

The arity of the introduced functions provides an important bound on the
treewidth of the network (and therefore computational
complexity). Thus, it is noteworthy that the base pair energy model
requires only binary functions; the stacking model, only quarternary
dependencies. This arity is increased in a few cases by the commonly
used Turner 2004 model~\citep{Turner2009} for encoding tabulated
special hairpin and interior loop contributions, which depend on up to
nine bases for the interior loops with a total of 5 unpaired bases
(``2x3'' interior loops)---all other energy contributions (like
dangling ends) still depend on at most four bases of the sequence.

\subsection{Extension to multidimensional Boltzmann sampling}\label{sec:multiBoltzmann}
The flexibility of our framework supports an advanced sampling technique, named multidimensional Boltzmann sampling~\citep{Bodini2010} to (probabilistically) enforce additional constraints.
This technique was previously used to control the \GCb-content~\citep{Waldispuehl2011,Reinharz2013} and dinucleotide content~\citep{Zhang2013} of sampled RNA sequences; it enables explicit control of the free energies $(\TargetE_1,\ldots,\TargetE_k)$ of the single targets. 

Multidimensional Boltzmann sampling requires the ability to \Def{sample from a weighted distribution} over the set of compatible sequences, where the probability of a sequence $S$ with free energies $(E_1,\ldots,E_k)$ for its target structures is
$\mathbb{P}(S\mid \pmb{\pi}) = \frac{\prod_{\ell=1}^{k} \pi_i^{-E_i}}{\partfun{\pmb{\pi}}},$
where $\pmb{\pi}:=(\pi_1\cdots\pi_k)$ is a vector of positive real-valued \Def{weights}, and $\partfun{\pmb{\pi}}$ is the weighted partition function. Such a distribution can be induced by a simple modification of the functions described in Sec.~\ref{sec:energy_models}, where any energy function $E(\val)$ for a structure $\ell$ is replaced by $E'(\val):= \ln(\pi_\ell)\, E(\val)/\beta$. The probability of a sequence $S$ is thus proportional to
$ \prod_{\ell=1}^{k} e^{-\ln(\pi_\ell)\, E_i} = \prod_{\ell=1}^{k} \pi_i^{-E_i}. $

One then needs to \Def{learn a weights vector} $\pmb{\pi}$ such that, on average, the targeted energies are achieved by a random sequences in the weighted distribution, \emph{i.e.} such that  $\mathbb{E}(E_\ell(S)\mid \pmb{\pi})=\TargetE_\ell$,  $\forall\ell\in[1,k]$.
The expected value of $E_\ell$ is always decreasing for increasing weights $\pi_\ell$ (see Supp. Mat.~\ref{sec:weight_derivatives}). More generally, computing a suitable $\pmb{\pi}$ can be restated as a convex optimization problem, and be efficiently solved using a wide array of methods~\citep{Denise2010,Bendkowski2017}.
In practice, we use a simple heuristics which starts from an initial weight vector $\pmb{\pi}^{[0]}:=(e^\beta,\dots,e^\beta)$ and, at each iteration, generates a sample $\mathcal{S}$ of sequences. The expected value of an energy $E_\ell$ is estimated as $\hat\mu_\ell(\mathcal{S}) = \sum_{S\in\mathcal{S}}E_\ell(S)/|\mathcal{S}|$, and the weights are updated at the $t$-th iteration by 
$\pi_\ell^{[t+1]} = \pi_\ell^{[t]}\cdot \gamma^{\hat\mu_\ell(\mathcal{S})-\TargetE_\ell}$. In practice, the constant $\gamma>1$ is chosen empirically to achieve effective optimization.
While heuristic in nature, this basic iteration was elected in our initial version of \ourprog{} because of its reasonable empirical behavior (choosing $\gamma=1.2$). 

A further \Def{rejection step} is applied to the generated structures to retain only sequences whose energy for each structure $R_\ell$ belongs to $[\TargetE_\ell\cdot(1-\varepsilon),\TargetE_\ell\cdot(1+\varepsilon)]$, for $\varepsilon\ge 0$ some predefined \Def{tolerance}. The rejection approach is justified by the following considerations:
i) \emph{Enacting an exact control over the energies would  be technically hard and costly.} Indeed, controlling the energies through dynamic programming would require explicit convolution products, generalizing~\citet{Cupal1996}, inducing additional $\Theta(n^{2k})$ time and $\Theta(n^k)$ space overheads;
ii) \emph{Induced distributions are typically concentrated.} Intuitively, unless sequences are fully constrained individual energy terms are independent enough so that their sum is concentrated around its mean -- the targeted energy (cf Fig.~\ref{fig:energydist}).
For base pair-based energy models and special base pair dependency graph
(paths, cycles\ldots) this property rigorously follows from analytic
combinatorics, see \citet{Bender1983} and
\citet{Drmota1997}. In such cases, the expected number of
rejections before reaching the targeted energies remains constant when
$\varepsilon\ge 1/\sqrt{n}$, and $\Theta(n^{k/2})$ when
$\varepsilon=0$. The \GCb-content of designs can also be controlled,
jointly and in a similar way, as done in
\Software{IncaRNAtion}~\citep{Reinharz2013}.

\section{Results}\label{sec:results}

\subsection{Targeting Turner energies and GC-content}
We implemented the Boltzmann sampling approach (Algorithms
\ref{alg:pf} and \ref{alg:sampling}), performing sampling for given
target structures and weights $\pi_1,\dots,\pi_n$; moreover on top,
multi-dimensional Boltzmann sampling (see
Section~\ref{sec:multiBoltzmann}) to target specific energies and
\GCb-content.  Our tool \ourprog{} evaluates energies according to the base
pair energy model or the stacking energy model, using parameters which
were trained to approximate Turner energies
(Supp. Mat.~\ref{appsec:modelparameters}).
%
%

To capture the realistic Turner model $E_{\mathcal{T}}$, we exploit its very good correlation (supp. Fig~\ref{fig:training-cor}) with our simple stacking-based model $E_{\mathcal{S}}$. Namely, we observe a structure-specific affine dependency between these Turner and stacking energy models, so that $E_{\mathcal{T}}(S_\ell;R_\ell) \approx \gamma_\ell\, E_{\mathcal{S};R_\ell}(S_\ell) + \delta_\ell$ for each structure $R_\ell$. We learn the $(\gamma_\ell,\delta_\ell)$ parameters from a set of sequences generated with homogenous weights $w=e^{\beta}$, tuning only the \GCb-content to a predetermined target frequency.  Finally, we adjust the targeted energy of our stacking model to $E_{\mathcal{S}}^{\star} = (E_{\mathcal{T}}^{\star}- \delta_\ell)/\gamma_\ell$.

To illustrate our above strategy, we sampled $n=1\,000$ sequences
targeting ${\sf GC}\%=0.5$ and different Turner energies for the three
structure targets of an example instance. Fig.~\ref{fig:energydist}A
illustrates how the Turner energy distributions of the shown target
structures can be accurately shifted to prescribed target energies;
for comparison, we plot the respective energy distribution of
uniformly and Boltzmann sampled sequences. Fig.~\ref{fig:energydist}B
summarizes the targeting accuracy for the single structure energies
over a larger set of instances from the \texttt{Modena} benchmark. We
emphasize that only the simple stacking energies are directly targeted
(at $\pm$10\% tolerance, expect for a few hard instances). Due to our
linear adjustment, we finally achieve well-controlled energies in the
realistic Turner energy model.
\begin{figure*}[t]
  \begin{center}
    {\sf \bfseries A}\includegraphicstop[width=0.57\textwidth]{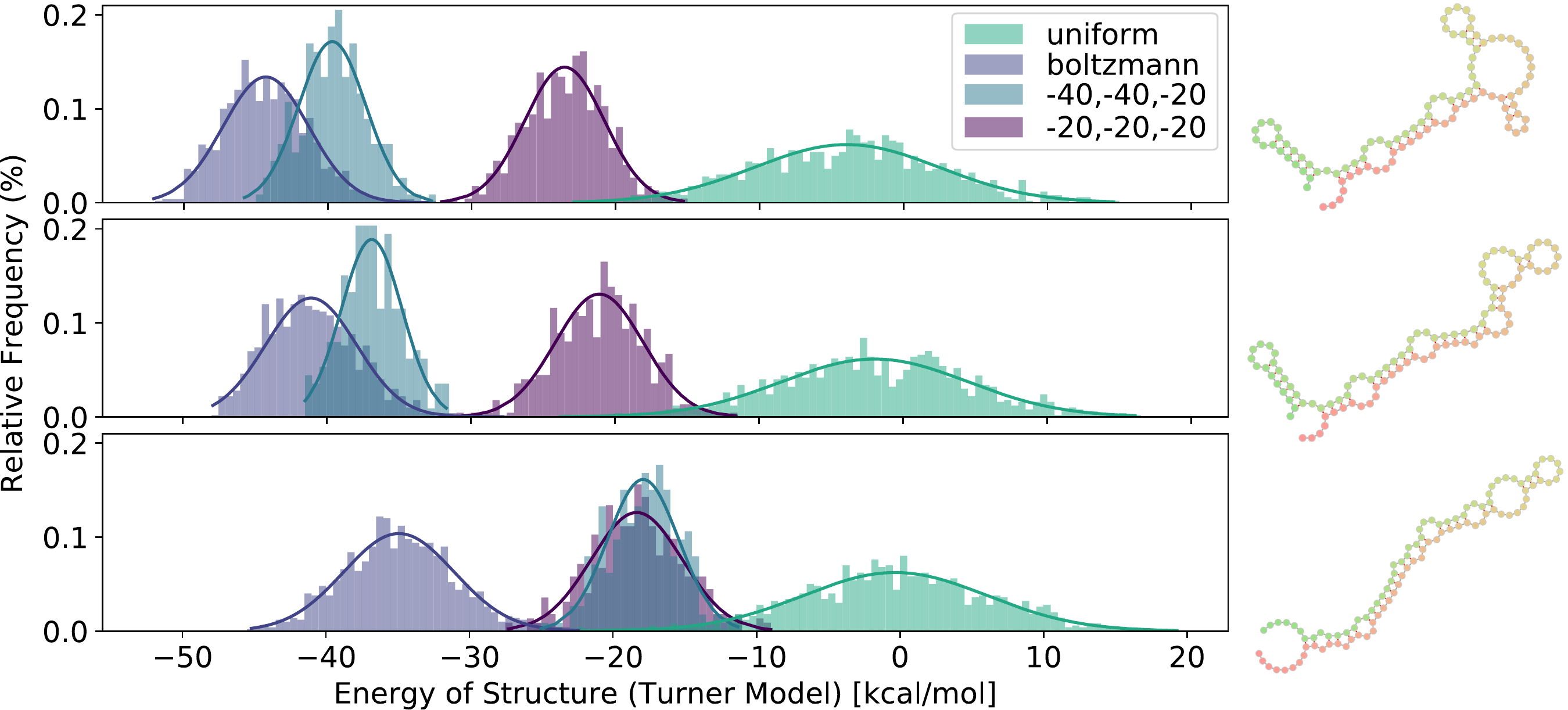}
    {\sf \bfseries B}\includegraphicstop[width=0.32\textwidth]{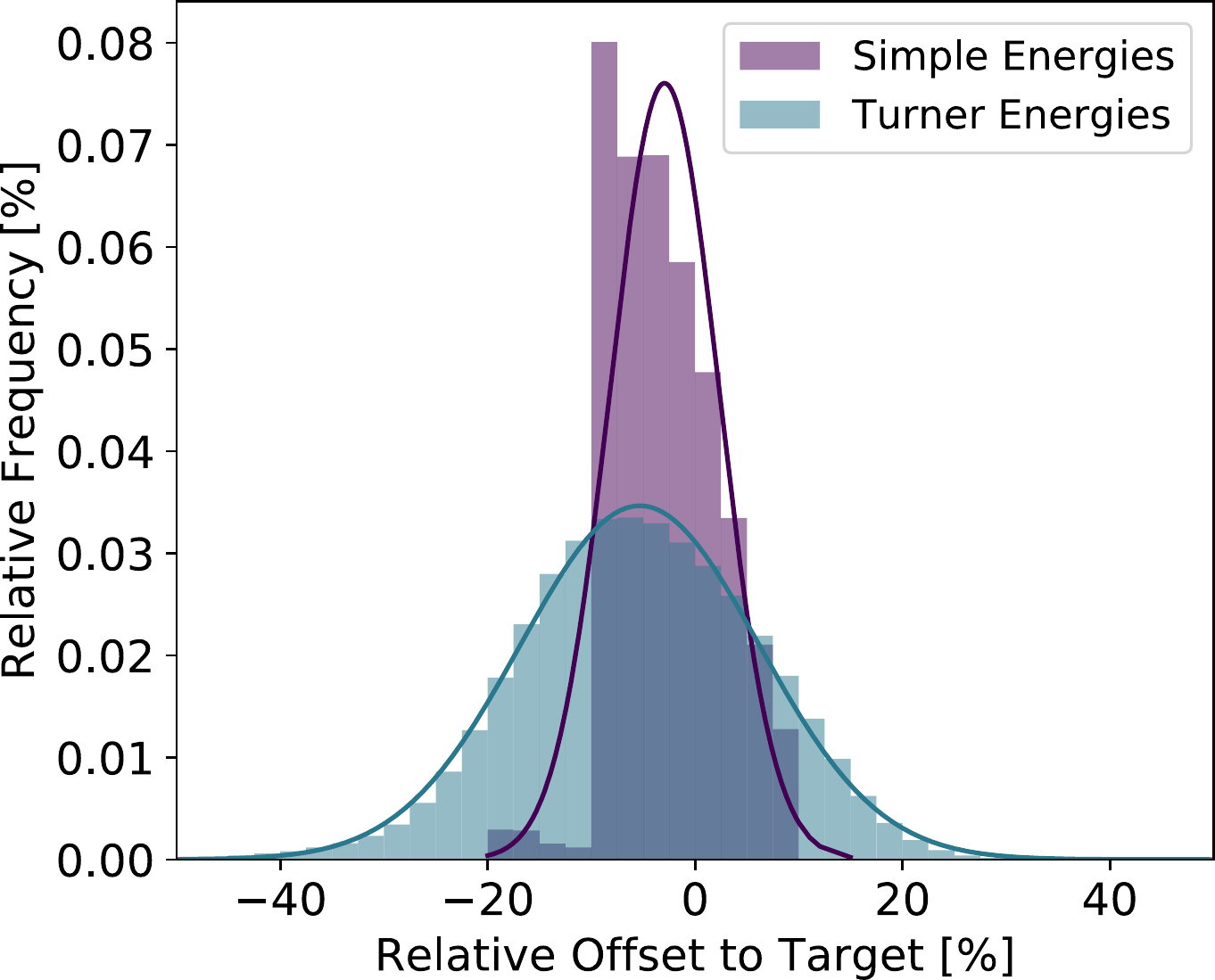}
  \end{center}
  \caption{%
  	Targeting specific energies using multi-dimensional Boltzmann sampling. (A) Turner energy distributions for three target structures (right) while targeting $(-40,-40,-20)$ and $(-20,-20,-20)$ free energy vectors. For comparison, we show the distribution of uniform and Boltzmann samples; respectively associated with homogenous weights $1$ and $e^\beta$. (B) Accuracy of targeting over all target structures of the \Software{Modena} benchmark instances \texttt{RNAtabupath} (2str), \texttt{3str} and \texttt{4str}. Shown is the relative deviation of targeted and achieved energies in the simple stacking model and the Turner model.
}
  \label{fig:energydist}
\end{figure*}
\subsection{Generating high-quality seeds for further optimization}
We empirically study the capacity of \ourprog{} to improve the generation of seed sequences for subsequent local optimizations, showing an important application in biologically relevant scenario.

For a multi-target design instance, individual reasonable target
energies are estimated from averaging the energy over samples at
relatively high weights $e^{\beta}$, setting the targeted \GCb-content
to 50\%. For each instance of the \Software{Modena}
benchmark~\citep{Taneda2015}, we estimated such target energies and
subsequently generated $1\,000$ seed sequences at these energies.
Our procedure is tailored to produce sequences with similar Turner
energy that favor the stability of the target structures having
moderate \GCb-content; all these properties are desirable objectives
for RNA design.  All sampled sequences are evaluated under the
multi-stable design objective function of \citet{Hammer2017}:
\begin{align}
  \label{eq:blueprintobjective}
    f(S)= \quad & \frac{1}{k} \sum_{\ell=1}^{k} (E(S, R_\ell) - G(S))\notag\\
   & +\ 0.5\frac{1}{\binom{k}{2}} \sum\limits_{1\leq\ell<j\leq k}|E(S,R_\ell) - E(S,R_j)|.
\end{align}

Using our strategy, we generated at least 1\,000 seeds per instance of
the subsets \texttt{RNAtabupath} (2str), \texttt{3str}, and \texttt{4str}, of
the \Software{Modena} benchmark set with 2,3, and 4 structures. The
results are compared, in terms of their value for the objective
function of Eq.~\eqref{eq:blueprintobjective}, to seed sequences
uniformly sampled using \Software{RNAblueprint}~\citep{Hammer2017}
(Table~\ref{tab:benchmark-results}, seeds). Our first analysis reveals
that Boltzmann-sampled sequences constitute better seeds, associated
with better values ($\sim\! 69\%$ improvement on average), compared to
uniform seeds for \emph{all instances} of our benchmark
(Supp. Mat.~\ref{sec:validity}).

Moreover, for each generated seed sequence, we optimized the objective function $(\ref{eq:blueprintobjective})$ through an adaptive  greedy walk consisting of 500 steps. At each step, we resampled (uniformly random) the positions of a randomly selected component in the base pair dependency graph, and accepted the modification only if it resulted in a gain, as described in \Software{RNAblueprint}~\citep{Hammer2017}.
Once again, for all instances, we observe a positive improvement of
the quality of Boltzmann designs over uniform ones, suggesting a
superior potential for optimization ($\sim 32\%$ avg improvement,
Table~\ref{tab:benchmark-results}, optimized). Notably, our subsequent
optimization runs, which are directly inherited from the uniformly
sampling RNABlueprint software, partially level the advantages of
Boltzmann sampling. In future work, we hope to improve this aspect by
exploiting Boltzmann sampling even during the optimization run.

\begin{table}[t]
\centering
\medskip
\begin{tabular}{@{}>{\bf}l@{\quad}>{\tt}l@{\quad}@{\quad}c@{\quad}c@{\quad}c@{}}
             &   \textbf{\textrm{Dataset}}   & {\bfseries\ourprog{}} & \textbf{Uniform} & \textbf{Improvement} \\\toprule
  Seeds      & 2str & 21.67 ($\pm$4.38) & 37.74 ($\pm$6.45) & 73\%\\
             & 3str & 18.09 ($\pm$3.98) & 30.49 ($\pm$5.41) & 71\%\\
             & 4str & 19.94 ($\pm$3.84) & 32.29 ($\pm$5.24) & 63\%\\\midrule

  Optimized  & 2str & 5.84 ($\pm$1.31) & 7.95 ($\pm$1.76) & 28\%\\
             & 3str & 5.08 ($\pm$1.10) & 7.04 ($\pm$1.52) & 31\%\\
             & 4str & 8.77($\pm$1.48) & 13.13 ($\pm$2.13) & 37\% \\ \bottomrule
\end{tabular}\\[1em]

\caption{Comparison of the multi-stable design objective function (
  Eq.~(\ref{eq:blueprintobjective}) of sequences sampled by \ourprog{}
  vs.\ uniform samples for three benchmark sets; before (seeds) and
  after optimization (optimized). We report the average scores with
  standard deviations. Smaller scores are better; the final column
  reports our improvements over uniform sampling.}
\label{tab:benchmark-results}
\end{table}

\section{Counting valid designs is \#{\sf P}-hard}\label{sec:counting}

Finally, we turn to the complexity of $\NumDesign(G)$, the problem of computing the number of valid sequences for a given compatibility graph $G=(V,E)$. Note that this problem corresponds to the partition function problem in a simple $(0/\infty)$-valued base pair model, with $\beta\neq 1$. As previously noted~\citep{Flamm2001}, a set of target structures admits a valid design iff its compatibility graph is bipartite, which can be tested in linear time.
Moreover, without $\{\Gb, \Ub\}$ base pairs, any connected component  $C\in {\rm CC}(G)$ is entirely determined by the assignment of a single of its nucleotides. The number of valid designs is thus simply $4^{\#{\rm CC}(G)}$, where $\#{\rm CC}(G)$ is the \Def{number of connected components}.

The introduction of $\{\Gb, \Ub\}$ base pairs radically changes this picture, and we show that valid designs for a set of structures cannot be counted in polynomial time, unless ${\sf\# P}={\sf FP}$. The latter equality would, in particular, imply the classic ${\sf P}={\sf NP}$, and solving $\NumDesign(G)$ is polynomial time is therefore probably difficult. 

To establish that claim, we consider instances $G=(V_1\cup V_2, E)$ that are connected and bipartite ($(V_1\times V_2) \cap E = \varnothing$), noting that hardness over restricted instances implies hardness of the general problem. Moreover remark that, as observed in Subsec.~\ref{sec:complexity}, assigning a nucleotide to a position $u\in V$ constrains the parity ($\{\Ab,\Gb\}$ or $\{\Cb,\Ub\}$) of all positions in the connected component of $u$. For this reason, we restrict our attention to the counting of valid designs \emph{up to trivial} $(\Ab\leftrightarrow \Cb/\Gb\leftrightarrow \Ub)$ symmetry, by constraining the positions in $V_{1}$ (resp. $V_2$) to only $\Ab$ and $\Gb$ (resp. $\Cb$ and $\Ub$). Let $\Design{G}$ denote the subset of all designs for $G$ under this constraint, noting that $\NumDesign(G) = 2\cdot|\Design{G}|$.

We remind that an \Def{independent set} of $G=(V,E)$ is a subset $V'\subseteq V$ of nodes that are not connected by any edge in $E$. Let $\IS{G}$ denote the set of all independent sets in $G$. 

\begin{proposition}
 $\Design{G}$ is in bijection with $\IS{G}$.
\end{proposition}
\begin{proof}
Consider the function $\Psi: \Design{G} \to \IS{G}$ defined by $ \Psi(f) := \left\{v\in V\mid f(v)\in\{\Ab,\Cb\}\right\}.$

Let us establish the injectivity of  $\Psi$, i.e. that $\Psi(f)\neq\Psi(f')$ for all $f\neq f'$.
If $f\neq f'$, then there exists a node $v\in V$ such that $f(v)\neq f'(v)$. 
Assume that $v\in V_1$, and remind that the only nucleotides allowed in $V_1$ are $\Ab$ and $\Gb$. Since $f(v)\neq f'(v)$, then we have $\{f(v),f'(v)\}=\{\Ab,\Gb\}$
and we conclude that $\Psi(f)$ differs from $\Psi(f')$ at least with respect to its inclusion/exclusion of $v$.

We turn to the surjectivity of $\Psi$, i.e. the existence of a preimage for each element $S\in \IS{G}$. Let us consider the function $f$ defined as
\begin{align}
 &\forall v_1\in V_1,v_2\in V_2:\notag\\& f(v_1) = \begin{cases} \Ab & \text{if }v_1\in S\\ \Gb & \text{if }v_1\notin S\end{cases} \text{\quad and\quad} 
 f(v_2) = \begin{cases} \Cb & \text{if }v_2\in S\\ \Ub & \text{if }v_2\notin S\end{cases}
\end{align}
One easily verifies that $\Psi(f) = S$. 

We are then left to determine if $f$ is a valid design for $G$, i.e. if for each $(v, v') \in E$ one has  $\{f(v),f(v')\}\in \B.$ Since $G$ is bipartite, any edge in $E$ involves two nodes $v_1\in V_1$ and $v_2\in V_2$. Remark that, among all the possible assignments $f(v_1)$ and $f(v_2)$, the only invalid combination of nucleotides is $(f(v_1),f(v_2)) = (\Ab,\Cb)$. However, such nucleotides are assigned to positions that are in the independent set $S$, and therefore cannot be adjacent. We conclude that $\Psi$ is surjective, and thus bijective.
\end{proof}

Now we can build on the connection between the two problems to obtain complexity results for \NumDesign. Counting independent sets in bipartite graphs ($\#{\sf BIS}$) is indeed a well-studied \#{\sf P}-hard problem~\citep{Ge2012}, from which we immediately conclude:
\begin{corollary}
  \NumDesign is $\#{\sf P}$-hard
\end{corollary}
\begin{proof}
  Note that $\#{\sf BIS}$ is also \#{\sf P}-hard on connected graphs, as the number of independent sets for a disconnected graph $G$ is given by $|\IS{G}|=\prod_{cc\in CC(G)} |\IS{cc}|$. Thus any efficient algorithm for $\#{\sf BIS}$ on connected instances provides an efficient algorithm for general graphs.
  
  Let us now hypothesize the existence of a polynomial-time algorithm $\mathcal{A}$ for \NumDesign over strongly-connected graphs $G$. Consider the (polynomial-time) algorithm $\mathcal{A}'$ that first executes $\mathcal{A}$ on $G$ to produce $\NumDesign(G)$, and returns $\NumDesign(G)/2=|\Design{G}| = |\IS{G}|$. Clearly $\mathcal{A}'$ solves $\#{\sf BIS}$ in polynomial-time. This means that $\NumDesign$  is at least as hard as $\#{\sf BIS}$, thus does not admit a polynomial time exact algorithm unless $\#{\sf P}={\sf FP}$.
\end{proof}

\section{Conclusion}
Motivated by the---here established---hardness of the problem, we
introduced a general framework and efficient algorithms for design of
RNA sequences to multiple target structures with fine-tuned
properties. Our method combines an FPT stochastic sampling algorithm
with multi-dimensional Boltzmann sampling over distributions
controlled by expressive RNA energy models.  Compared to the
previously best available sampling method (uniform sampling), this
approach generated significantly better seed sequences for instances
of a common multi-target design benchmark.

The presented method enables new possibilities for sequence generation
in the field of RNA sequence design by enforcing additional
constraints, like the \GCb-content, while controlling the energy of
multiple target structures. Thus it presents a major advance over
previously applied ad-hoc sampling and even efficient uniform sampling
strategies. We have shown the practicality of such controlled sequence
generation and studied its use for multi-target RNA design. Moreover,
our framework is equipped to include more complex sequence
constraints, including mandatory/forbidden motifs at specific
positions or anywhere in the designed sequences. Furthermore, the
method can support even negative design principles, for instance by
penalizing a set of alternative helices/structures. Based on this
generality, we envision our framework---in extensions, which still
require delicate elaboration---to support various complex RNA design
scenarios far beyond the sketched applications.

\Final{
Extensions:
\begin{enumerate}
  \item What if the input graph is not bipartite? \\
  Extract a minimal subset of edges such that new graph is bipartite!\\
  $\to{}$ Graph bipartization (cf implementation by Hüffner, JGAA 2009)
  \item Is the optimization version in the \emph{weighted base-pair model} already {\sf NP}-hard?
  \item Does it help to know that there are cliques of size $k$ during the tree decomposition?\\
  If anything, it provides a lower bound on tree decomposition (helps branch and bound algorithms)
  \item What if some degree of freedom is allowed in the structure definition?\\ Can we still use tree decomposition?\\
  This looks like a very reasonable first step towards the definition of an {\em optimal} DP decomposition for RNA folding (with Pseudoknots) or RNA-RNA interactions\ldots
  \item Structures can also be penalized, but they must be specified explicitly.\\
  However, any competing structure must consist of helices, and it is possible to penalize each of the $\Theta(n^2)$ helices of any given length $L$. What would be the impact on treewidth?
  \item More general energy  models $\to$ Turner. (Pseudoknots)
  For this description, we assume that all input structures are non-crossing, while the framework would otherwise support crossing input structures as well as according energy models (e.g. the Hotknots model~\citep{Ren2005}) by considering special pseudoknot loops.
\end{enumerate}\label{lbl:theend}
}

\section*{Acknowledgements}
YP is supported
by the
{\em Agence
  Nationale de la Recherche} and the
FWF (ANR-14-CE34-0011; project RNALands).  SH is supported by the
German Federal Ministry of Education and Research (BMBF support code
031A538B; de.NBI: German Network for Bioinformatics Infrastructure)
and the
Future and Emerging Technologies programme
(FET-Open grant 323987; project RiboNets).
We thank Leonid Chindelevitch for suggesting using the stable sets analogy to optimize our FPT algorithm, and Arie Koster for practical recommendations
for computing tree decompositions.

\bibliographystyle{natbib}
\bibliography{biblio}

{\centering \relsize{+4}Supplementary material\\%
}
\relsize{+1}
\section{Approximate counting and random generation}
In fact, not only is $\#{\sf BIS}$ a reference problem in counting complexity, but it is also a landmark problem with respect to the complexity of approximate counting problems. In this context, it is the representative for a class of $\#{\sf BIS}$-hard problems~\citep{Bulatov2013} that are easier to approximate than $\# {\sf SAT}$, yet are widely believed not to admit any Fully Polynomial-time Randomized Approximation Scheme. Recent results reveal a surprising dichotomy in the behavior of $\#{\sf BIS}$: it admits a Fully Polynomial-Time Approximation Scheme (FPTAS) for graphs of max degree $\le 5$~\citep{Weitz2006}, but is as hard to approximate as the general $\#{\sf BIS}$ problem on graphs of degree $\ge 6$~\citep{Cai2016}. In other words, there is a clear threshold, in term of the max degree, separating (relatively) easy instances from really hard ones.

Additionally, let us note that, from the classic Vizing Theorem, any bipartite graph $G$ having maximum degree $\Delta$ can be decomposed in polynomial time in exactly $\Delta$ matchings. Any such matching can be reinterpreted as a secondary structure, possibly with crossing interactions (\Def{pseudoknots}). These results have two immediate consequences for the pseudoknotted version of the multiple design counting problem.
\begin{corollary}[as follows from~\citep{Weitz2006}]The number of designs compatible with $m\le 5$ pseudoknotted RNA structures can be approximated within any fixed ratio by a deterministic polynomial-time algorithm.
\end{corollary}
\begin{corollary}[as follows from~\citep{Cai2016}]
  As soon as the number of pseudoknotted RNA structures strictly exceeds $5$, \NumDesign is as hard to approximate as {\#{\sf BIS}}.
\end{corollary}

It is worth noting that the $\#{\sf P}$ hardness of \NumDesign does not immediately imply the hardness of generating a valid design uniformly at random, as demonstrated constructively by Jerrum, Valiant and Vazirani~\citep{Jerrum1986}. However, in the same work, the authors establish a strong connection between the complexity of approximate counting and the uniform random generation. Namely, they showed that, for problems associated with self-reducible relations, approximate counting is equally hard as (almost) uniform random generation. We conjecture that the (almost) uniform sampling of sequences from multiple structures with pseudoknots is in fact \#{\sf BIS}-hard as soon as the number of input structures strictly exceeds $5$, as indicated by~\citet{Goldberg2004}, motivating even further our parameterized approach.

\section{Tree decomposition for RNA design instances in practice}
\label{appsec:treedecomp}

For studying the typically expected treewidths and tree decomposition
run times in multi-target design instances, we consider five sets of
multi-target RNA design instances of different complexity. Our first
set consists of the Modena benchmark instances.

In addition, we generated four sets of instances of increasing
complexity. The instances of the sets RF3, RF4, RF5, and RF6, each
respectively specify 3,4,5, and 6 target structures for sequence
length 100.  For each instance (100 instances per set), we generated a
set of $k$ ($k=3,\dots,6$) compatible structures as follows
\begin{itemize}
\item Generate a random sequence of length 100;
\item Compute its minimum free energy structure (ViennaRNA package);
\item Add the new structure to the instances if the resulting base pair dependency graph is bipartite;
\item Repeat until $k$ structures are collected.
\end{itemize}
For each instance, we generated the dependency graphs in the base pair
model and in the stacking model. Then, we performed tree decomposition
(using strategy ``GreedyFillIn'' of LibTW~\citep{Dijk2006}) on each dependency
graph. The obtained treewidths are reported in
Fig.~\ref{fig:td-widths}, while Fig.~\ref{fig:td-times} shows the
corresponding run-times of the tree decomposition.
Fig.~\ref{fig:td-example} shows the tree decompositions for an example
instance from set RF3.

\begin{figure}
  \centering\includegraphics[width=0.9\textwidth]{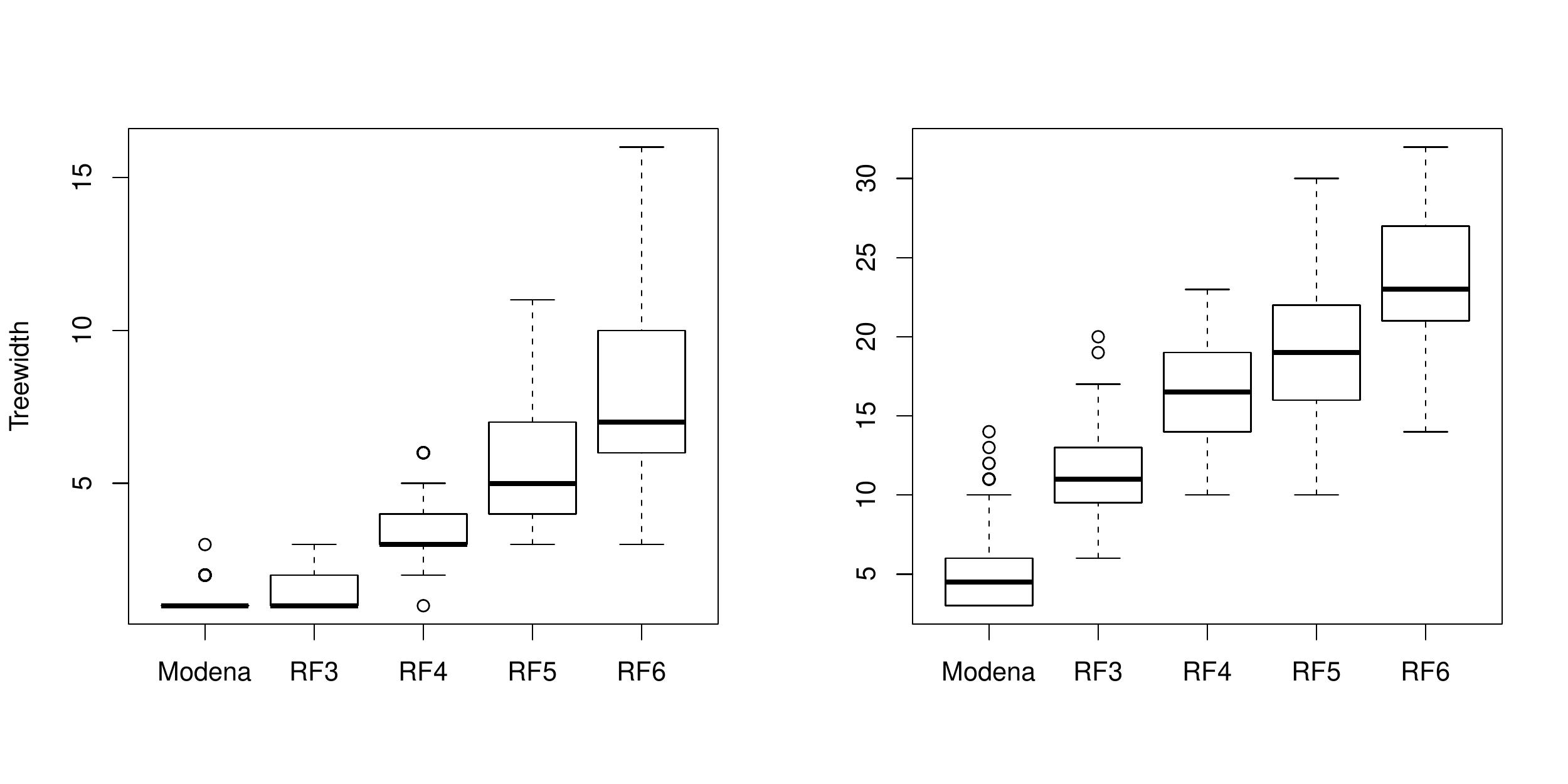}
  \caption{Treewidths for multi-target RNA design instances of
    different complexity. Distributions of treewidths shown as boxplots for the base pair (left) and stacking model (right).}
  \label{fig:td-widths}
\end{figure}

\begin{figure}
  \centering
  \includegraphics[width=0.9\textwidth]{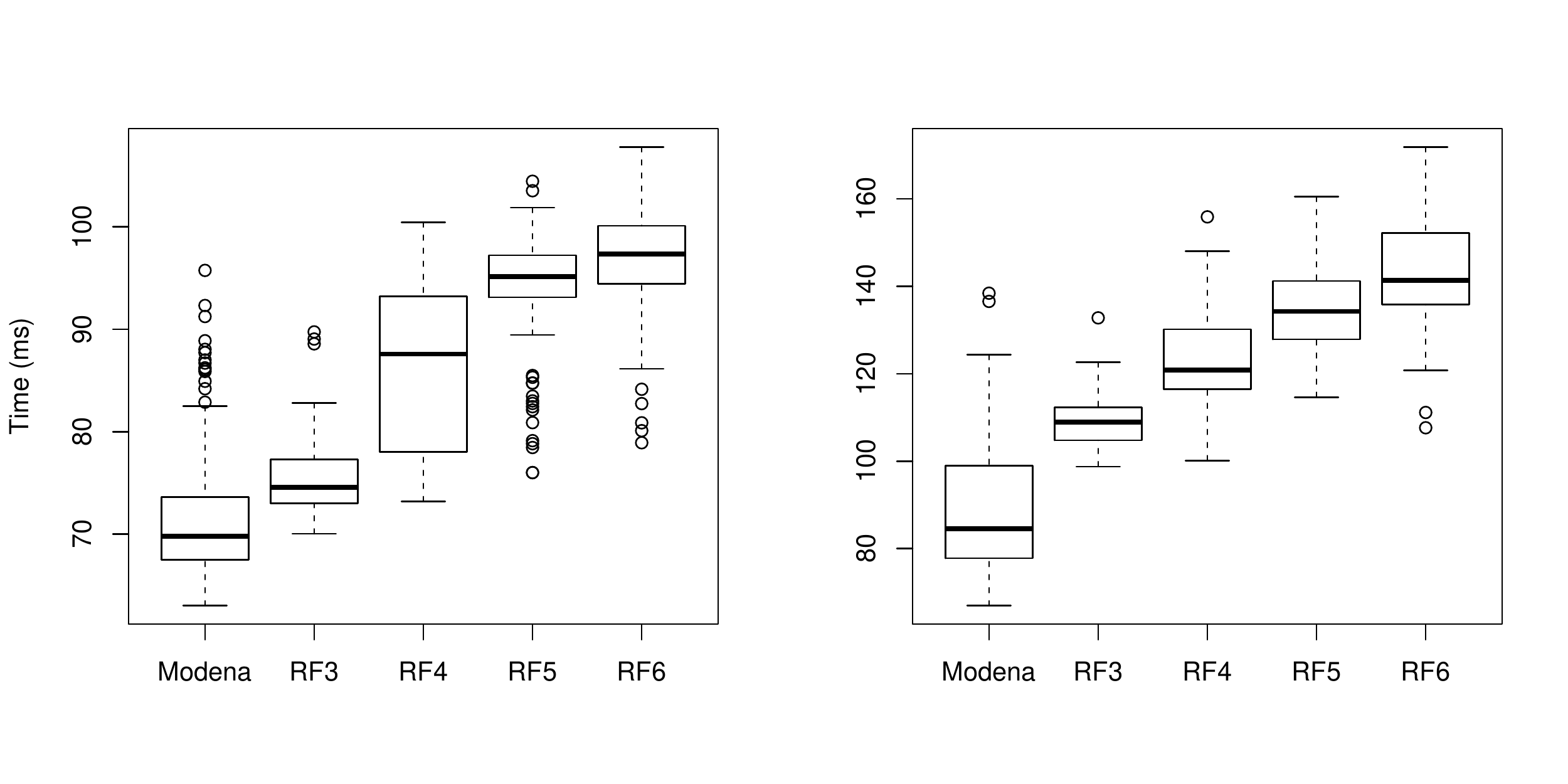}
  \caption{Computation time of tree decompositions for
    multi-target RNA design instances of different complexity.
    Distributions of times (in ms/instance) shown as boxplots for the base pair (left) and stacking model (right).}
  \label{fig:td-times}
\end{figure}

\begin{figure}[t!]
  \centering
\begin{verbatim}
((((.((....)).)))).((.(((.((((.....(((..((((((.((..(((.(.....).)))..)).)).))))..)))..)))).))).))....
..(((((.....(((.(((((((.....))))..))).))).....)))))..((((((((((...))).)....))))))...((((((....))))))
......(((((.....(((...(((.((.((.(((....((......))...))).)).)))))..))).............))))).((((...)))).
\end{verbatim}
  \includegraphicstop[width=0.6\textwidth]{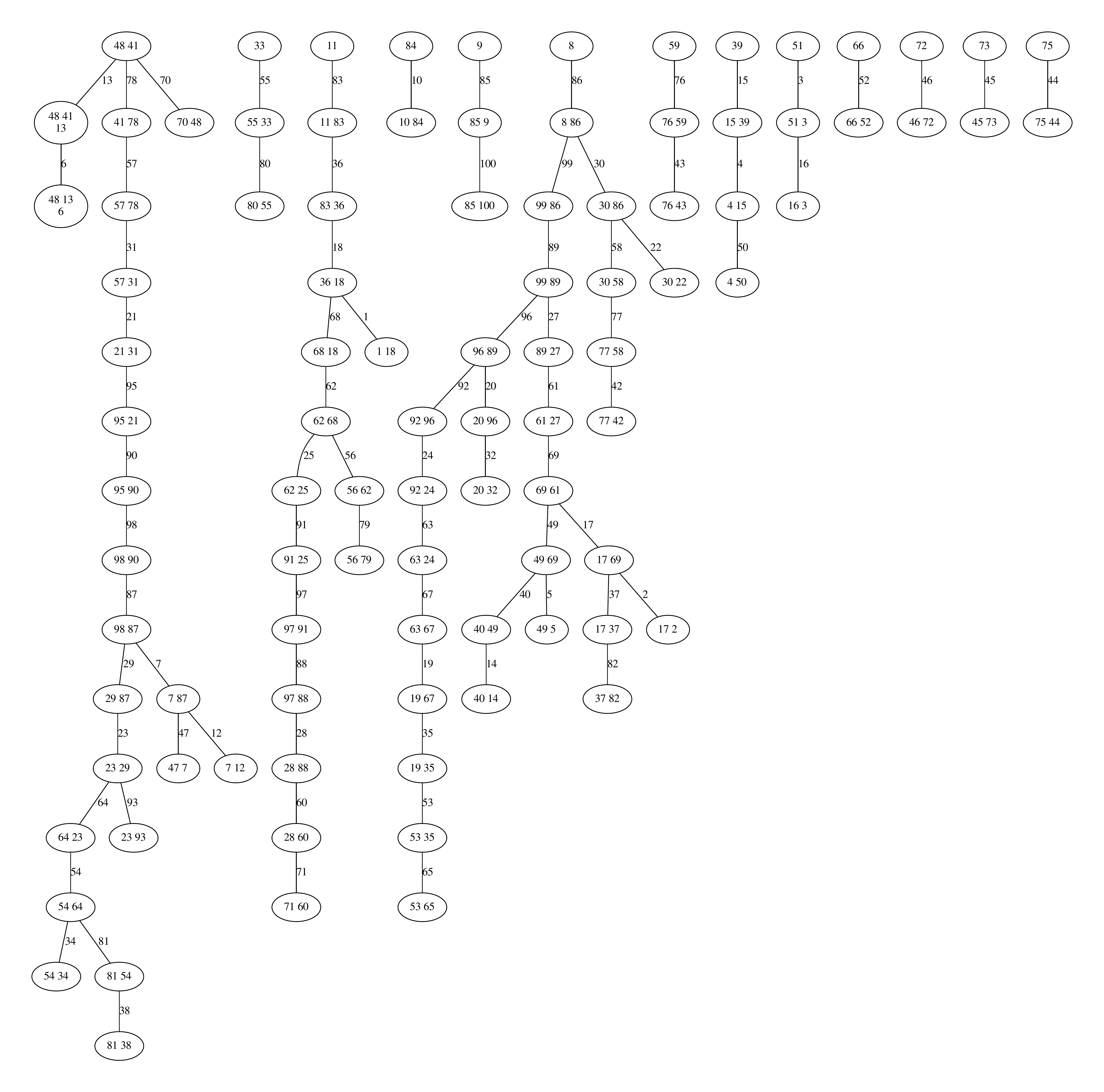}%
  \includegraphicstop[width=0.35\textwidth]{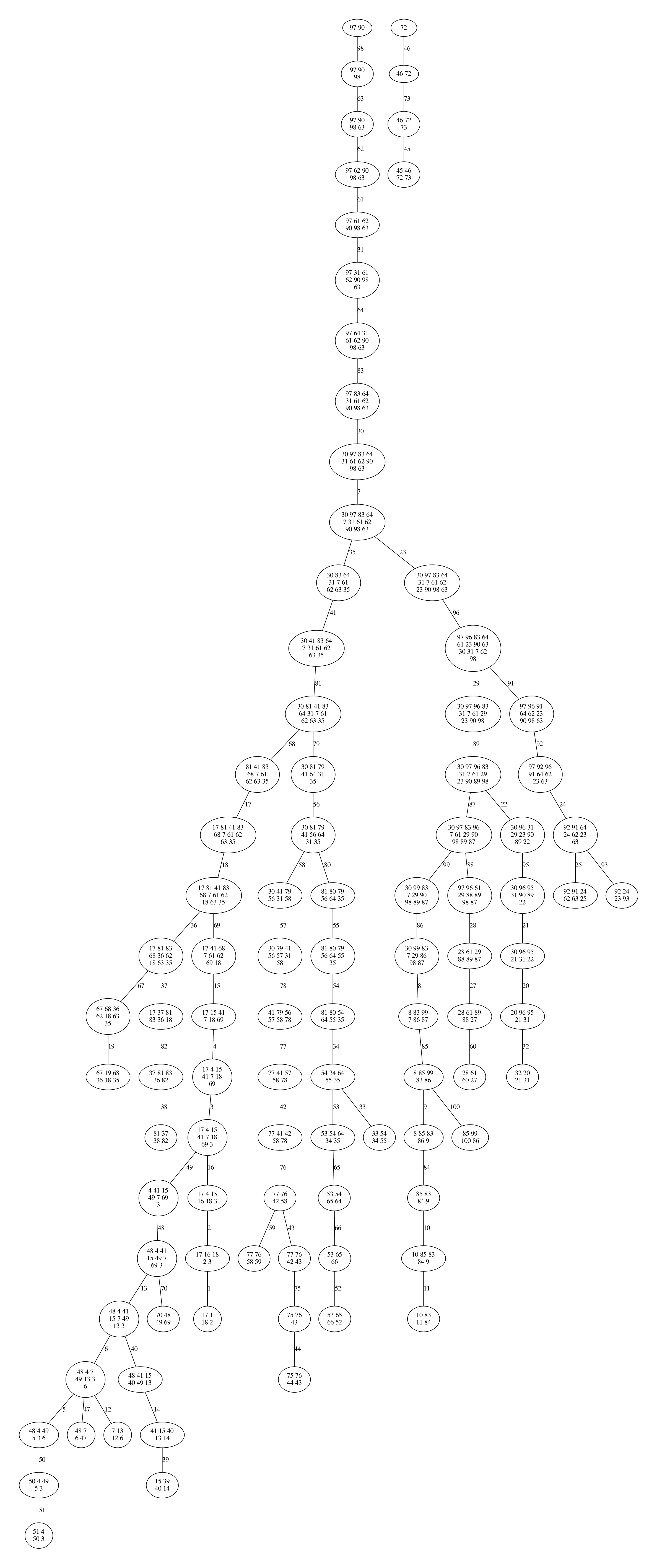}
  \caption{Example instance from RF3 (top) with its tree decompositions in the base pair (left) and stacking model (right). The respective treewidths are 2 and 12.}
  \label{fig:td-example}
\end{figure}

\section{Expressing higher arity functions as cliques in the dependency graph}
\label{appsec:dependency-cliques}

Our implementation relies on external tools for the tree decomposition. Therefore it is not trivial that the weight functions (i.e. terms of the energy function) can be captured within the tree decompositions returned by a specific tool. In other words, it is not immediately clear how one can construct a cluster tree from an arbitrary tree decomposition for a given network.

Crucially, one can show that, as long as the network the arguments $\dep(f)$ of a function $f$ are materialized by a clique within the network $\network$, then there exists at least one node in the tree decomposition that contains $\dep(f)$, thus allowing the evaluation of $f$.
\begin{lemma}\label{lem:cliques}
Let $G=(V,E)$ be an undirected graph, and  $T,\chi$ be a tree decomposition for $G$. For each clique $C\subset V$, there exists a node $n\in T$ such that $C\subset\chi(v)$.
\end{lemma}
\begin{proof}
  \newcommand{\TransVars}[1]{{\rm tverts}(#1)}
  \newcommand{\Children}[1]{{\rm Children}(#1)}
  For the sake of simplicity, let us consider $T$ as rooted on an arbitrary node, inducing an orientation, and denote by $\TransVars{n}\in V$ the set of vertices found in a node $n$ or its descendants, namely
    $$\TransVars{n} = \chi(n) \bigcup_{n'\in\Children{n}} \TransVars{n'}.$$

  Now consider (one of) the deepest node(s) $n\in T$, such that $C\subset\TransVars{v}$. We are going to prove that $C\subset\chi(n)$, using contradiction, by showing that $C\not\subset\chi(n)$ implies that $T$ is not a tree decomposition.

  Indeed, let us assume that $C\not\subset\chi(n)$, then there exists a node $v'\notin\chi(n)$ whose presence in $\TransVars{n}$ stems from its presence in some descendant of $n$. Let us denote by $n'\neq n$ the closest descendant of $n$ such that $v'\in\chi(n')$. Now, consider a node $v\in C$ such that $v\in \TransVars{v'}$ and $v\notin \TransVars{v'}$. Such a node always exists, otherwise $n'$, and not $n$, would be the deepest node such that $C\subset \TransVars{v'}$. From the definition of the tree decomposition, we know that neither $n'$ nor its descendants may contain $v$. On the other hand, none of the parents or siblings of $n'$ may contain $v'$. It follows that is no node $n''\in T$ whose vertex set $\chi(v'')$ includes, at the same time, $v$ and $v'$. Since both $v$ and $v'$ belong to a clique, one has $\{v,v'\}\in E$. The absence of a node in $T$ capturing an edge in $E$ contradicts our initial assumption that $T$ is tree decomposition for $G$.

  We conclude that $n$ is such that $C\not\subset\chi(n)$.
\end{proof}

This result implies that classic tree decomposition algorithms and, more importantly, implementations can be directly re-used to produce decompositions that capture energy models of arbitrary complexity, functions of arbitrary arity. Indeed, it suffices to add a clique, involving the parameters of the function, and Lemma~\ref{lem:cliques} guarantees that the tree decomposition will feature one node to which the function can be associated.

\section{Correctness of the FPT partition function algorithm}
\label{appsec:correctness}

\begin{theorem}[Correctness of Alg.~\ref{alg:pf}]
  \label{the:pfalgo-correctness}
  As computed by Alg.~\ref{alg:pf} for cluster tree $(T,\chi,\phi)$,
  the messages $\Message{u}{v}$, for all edges $u\to{}v\in T$, yield
  the partition functions of subtree of $u$ for the partial sequences
  $\val\in\partseqs(\separator{u}{v})$, i.e. the messages satisfy
  \begin{equation}
   \evalfor{\Message{u}{v}}{\val} = \sum_{\val'\in\partseqs(\chi(T_r(u))-\separator{u}{v})} \quad
   \prod_{f\in\phi(T_r(u))} \exp(-\beta \evalfor{f}{\substitute{\val'}{\val}}),\label{eq:pfalgo-correct}
 \end{equation}
 where $\chi(T_r(u))$ denotes all $\chi$-assigned positions of nodes in $T_r(u)$;
 respectively $\phi(T_r(u))$; all $\phi$-assigned functions.
\end{theorem}

\begin{proof}
  Note that in more concise notation, Alg.~\ref{alg:pf} computes
  messages such that
\begin{equation}
  \evalfor{\Message{u}{v}}{\val} := \sum_{\val'\in\partseqs(\difference{u}{v})}\quad
  \prod_{f\in \phi(u) } \exp(-\beta \evalfor{f}{\substitute{\val'}{\val}}) \prod_{(w\to{}u) \in T} \evalfor{\Message{w}{u}}{\substitute{\val'}{\val}}.\label{eq:messages}
\end{equation}

Proof by induction on $T$. If $u$ is a leaf, $\chi(r) = \chi(T_r(u))$,
there are no messages sent to $u$, and $\phi(u) = \phi(T_r(u));$ implying
Eq.~$(\ref{eq:pfalgo-correct}).$
Otherwise, since the algorithm traverses edges in postorder, $u$ received from its children
$w_1,\dots,w_q$ the messages $\Message{w_1}{u}, \dots, \Message{w_q}{u}$, which satisfy Eq.~$(\ref{eq:pfalgo-correct})$ (induction hypothesis). Let $\val\in\partseqs(\separator{u}{v})$; then, $\evalfor{\Message{u}{v}}{\val}$ is computed by the algorithm according to Eq.~$(\ref{eq:messages})$. We rewrite as follows
{\small
\begin{align*}
  & \sum_{\val'\in\partseqs(\difference{u}{v})}\quad
    \prod_{f\in \phi(u) } \exp(-\beta \evalfor{f}{\substitute{\val'}{\val}})
    \prod_{(w\to{}u) \in T} \evalfor{\Message{w}{u}}{\substitute{\val'}{\val}}\\
  & = \sum_{\val'\in\partseqs(\chi(u)-\separator{u}{v})}\quad
    \prod_{f\in \phi(u) } \exp(-\beta \evalfor{f}{\substitute{\val'}{\val}})
    \prod_{i=1}^q \evalfor{\Message{w_i}{u}}{\substitute{\val'}{\val}}\\
  & =_{IH}
    \sum_{\val'\in\partseqs(\chi(u)-\separator{u}{v})}\quad
    \prod_{f\in \phi(u) } \exp(-\beta \evalfor{f}{\substitute{\val'}{\val}})
    \prod_{i=1}^q \sum_{\val''\in\partseqs(\chi(T_r(w_i))-\separator{w_i}{u})} \quad
    \prod_{f\in\phi(T_r(w_i))} \exp(-\beta \evalfor{f}{\substitute{\val''}{\substitute{\val'}{\val}}}) \\
& =_{*}
    \sum_{\val'\in\partseqs(\chi(u)-\separator{u}{v})}\quad
  \sum_{\val''\in\partseqs(\bigcup_{i=1}^q\chi(T_r(w_i))-\separator{w_i}{u})} \quad
  \prod_{f\in \phi(u) } \exp(-\beta \evalfor{f}{\substitute{\val'}{\val}})
  \prod_{f\in\phi(T_r(w_i))} \exp(-\beta \evalfor{f}{\substitute{\val''}{\val'} | \val}) \\
  & =_{(**)} \sum_{\val'\in\partseqs(\chi(T_r(u))-\separator{u}{v})}\quad
    \prod_{f\in \phi(T_r(u)) } \exp(-\beta \evalfor{f}{\substitute{\val'}{\val}})\\
\end{align*}}
To see (*) and (**), we observe:
\begin{itemize}
\item The sets $\chi(u)-\separator{u}{v}$ and
  $\chi(T_r(w_i))-\separator{w_i}{u}$ are all disjoint due to
  Def.~\ref{def:treedecomp}, property 2. First, this property implies
  that any shared position between the subtrees of $w_i$ and $w_j$
  must be in $\chi(w_i)$, $\chi(w_j)$ and $\chi(u)$, thus the
  positions of $\chi(T_r(w_i))-\separator{w_i}{u}$ are
  disjoint. Second, if a position $\chi(T_r(w_i))$ occurs in
  $\chi(u)$, it must occur in $\chi(w_i)$ and consequently in $\separator{u}{v}$.
\item The union of the sets $\chi(u)-\separator{u}{v}$ and
  $\chi(T_r(w_i))-\separator{w_i}{u}$ is $\chi(T_r(u))-\separator{u}{v}).$
\end{itemize}

\end{proof}

\section{General complexity of the partition function computation by cluster tree elimination and generation of samples}
\label{appsec:algcomplexity}

\begin{proposition}
  \label{prop:general-complexity}
Given a cluster tree $(T,\chi,\phi)$, $T=(V,E)$ of the RNA design network $\network=(\X,\F)$ with treewidth $w$ and maximum separator size $s$, computing the partition function by Algorithm~\ref{alg:pf} takes $O((|F|+|V|)\cdot 4^{w+1})$ time and $O(|V| 4^s)$ space (for storing all messages). The sampling step has time complexity of $O((|F|+|V|)\cdot 4^D)$ per sample.
\end{proposition}

\begin{proof}[Proposition~\ref{prop:general-complexity}]
Let $d_u$ denote the degree of node $u$ in $T$, $s_u$ is the size of the separator between
$u$ and its parent in $T$ rooted at $r$. For each node in the cluster tree, Alg.~\ref{alg:pf} computes one message by combining $(|\phi(u)|+d_u-1)$ functions, each time enumerating $4^{|\chi(u)|}$ combinations; Alg.~\ref{alg:sampling} computes $4^{|\chi{u}|-s_u}$ partition functions each time combining $(|\phi(u)|+d_u-1)$ functions.  $4^{|\chi(u)|}$ is bound by $4^{w+1}$ and $4^{|\chi{u}|-s_u}$ by $4^D$; moreover $\sum_{u\in V} (|\phi(u)|+d_u-1) = |F|+|V|-1$.
\end{proof}

\section{Exploiting constraint consistency to reduce the complexity}\label{sec:improvedComplexity}


While Alg.~\ref{alg:pf} computes messages values for
\emph{all} possible combinations of nucleotides for the positions 
in a node, 
we observe here that many such combinations are \emph{not} required for computing all relevant
partition functions. In particular, the algorithm can be restricted to
consider only \Def{valid} combinations, satisfying the (hard)
constraints induced by valid base pairs.

RNA design introduces constraints due to the requirement of canonical
(aka complementary) base pairing, which can be exploited in a
particularly simple and effective way to reduce the complexity.  
As previously noted by \citet{Flamm2001}, the base pair complementarity induces a
bi-partition of each connected component within the base pair
dependency graph, such that the nucleotides assigned to the two set of
nodes in the partition are restricted to values in $\{\Ab,\Gb\}$ and
$\{\Cb,\Ub\}$ respectively. We call a partial sequence
\Def{cc-valid}, iff its determined positions are consistent with such a separation
for all determined positions of the same connected component.


One can now modify Alg.~\ref{alg:pf}, on a tree decomposition
$(T,\chi)$, such that the message values
$\evalfor{\Message{u}{v}}{\val}$ are only computed for cc-valid
partial sequences $\val\in\partseqs(\separator{u}{v})$. Moreover, the
loop over $\val'\in\partseqs(\difference{u}{v})$ is restricted, such
that $\substitute{\val}{\val'}$ are cc-valid. Analogous restrictions are 
then implemented in the sampling algorithm Alg.~\ref{alg:sampling}.

The correctness of the modified algorithm follows from the same
induction argument, where the message computation over one node
evaluates messages from its children only at cc-valid partial
sequences. The result of this computation is a message, which corresponds 
to the partition function, restricted to cc-valid partial sequences. 
Since invalid sequences have infinite energy, they do not contribute to 
the partition function, and the partition function restricted to cc-valid sequences coincides 
with the initial one.

This restriction drastically improves the time
complexity. Indeed, for any given node $v$, the original algorithm sends message for 
 $\chi(v):=\substitute{\val}{\val'}$ such that
$\val\in\partseqs(\separator{u}{v})$ and
$\val'\in\partseqs(\difference{u}{v})$, while the modified algorithm
only considers cc-valid assignments for $\chi(v)$. Remark that, in any connected component $cc$, 
assigning some nucleotide to a position reduces cc-valid assignments to (at most) 
two alternatives for each of the $|cc|-1$ remaining positions. It follows that, for a node $v$ featuring positions from $\#cc(v)$ distinct connected components $\{cc_1,cc_2\ldots\}$ in the base pair dependency graph, the number of cc-valid assignments to positions in $\chi(v)$ is exactly
$$4^{\#cc(v)}\prod_{i=1}^{\#cc(v)} 2^{|cc_i|-1} = 2^{\#cc(v)}\, 2^{|\chi(v)|} \in \Theta(2^{\#cc}\, 2^{w+1}),$$
for a single node, where $\#cc$ is the total number of connected components in the base pair dependency graph, and $w$ is the tree-width of the tree decomposition $T$. Since the number of nodes in $T$ is in $\Theta(n)$, and the number of atomic energy contributions associated with $k$ structures is in $\mathcal{O}(k\,n)$, then the overall complexity grows like $\Theta(n\, k\, 2^{\#cc}\, 2^{w+1})$.

%
%

Finally, we remark that even stronger time savings could be possible
in practice, since cc-valid partial sequences can still violate
complementarity constraints, e.g. by assigning C and A to positions in
different sets of a partition, thus satisfying the bi-partition
constraints, where the positions base pair directly, rendering the
partial sequence invalid. Moreover, applications of the sampling
framework, can introduce additional constraints that further reduce
the number of valid partial subsequences. However, exploiting all such
constraints, in a complete and general way, would likely cause significant
implementation overhead, while not significantly improving
the asymptotic complexity.

\section{Parameters for the base pair and the stacking energy model}
\label{appsec:modelparameters}

We trained parameters for two RNA energy models to approximate the
Turner energy model, as implemented in the \Software{ViennaRNA}
package.  In the \Def{base pair model}, the total energy of a sequence $S$
for an RNA structure $R$ is given as sum of base pair energies, where
we consider six types of base pairs distinguishing by the bases $\Ab$-$\Ub$,
$\Cb$-$\Gb$ or $\Gb$-$\Ub$ (symmetrically) and between stacked and non-stacked base
pairs; here, we consider base pairs $(i,j)\in R$ \Def{stacked} iff
$(i+1,j-1)\in R$, otherwise \Def{non-stacked}. In the \Def{stacking model},
our features are defined by the stacks, i.e. pairs $(i,j)$ and
$(i+1,j-1)$ which both occur in $R$; we distinguish 18 types based on
$S_i,S_j,S_{i+1},S_{j-1}$ (i.e. all combinations that allow canonical
base pairs; the configurations $S_i,S_j,S_{i+1},S_{j-1}$ and
$S_{i+1},S_{j-1},S_j,S_i$ are symmetric).

Both models describe the energy assigned to a pair of sequence and
structure in linear dependency of the number of features and their
weights. We can thus train weights for linear predictors of the Turner
energy in both models.

For this purpose, we generated 5000 uniform random RNA sequences of
random lengths between 100 and 200. For each sample, we predict the
minimum free energy and corresponding structure using the ViennaRNA
package; then, we count the distinguished features (i.e., base pair or
stack types). The parameters are estimated fitting linear models
without intercept (\texttt{R} function \texttt{lm}). For both models,
\texttt{R} reports an adjusted R-squared value of 0.99. The resulting
parameters are reported in Tables~\ref{tab:basepairmodel} and
\ref{tab:stackingmodel}.

For validating the trained parameters, we generate a second
independent test set of random RNA sequences in the same
way. Fig.~\ref{fig:training-cor} shows correlation plots for the
trained parameters in the base pair and stacking models for predicting
the Turner energies in the test set with respective correlations of
$0.95$ an $0.94$.

\begin{figure}
  \centering
  A)\includegraphicstop[width=0.4\textwidth,trim=0 0 0 50,clip]{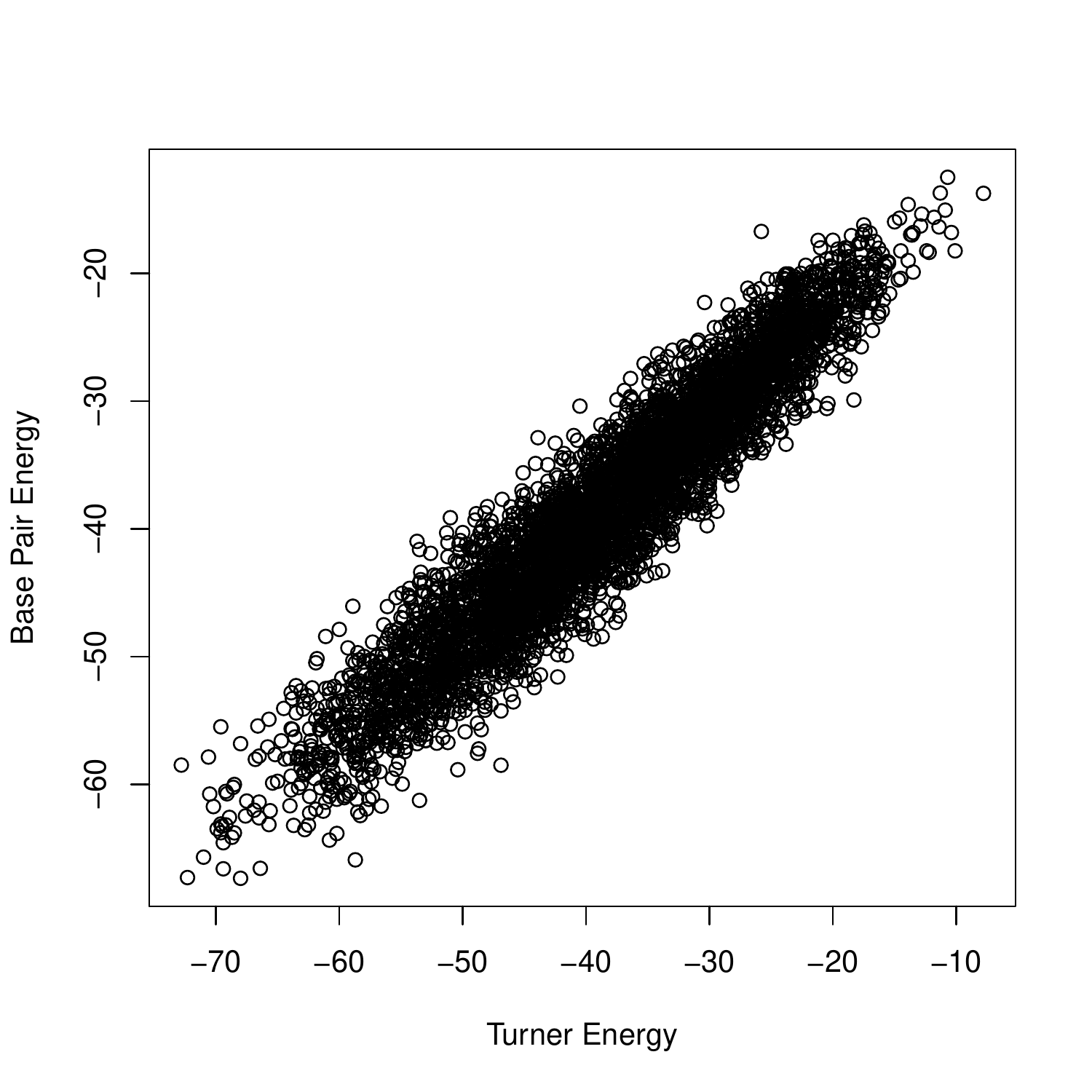}
  B)\includegraphicstop[width=0.4\textwidth,trim=0 0 0 50,clip]{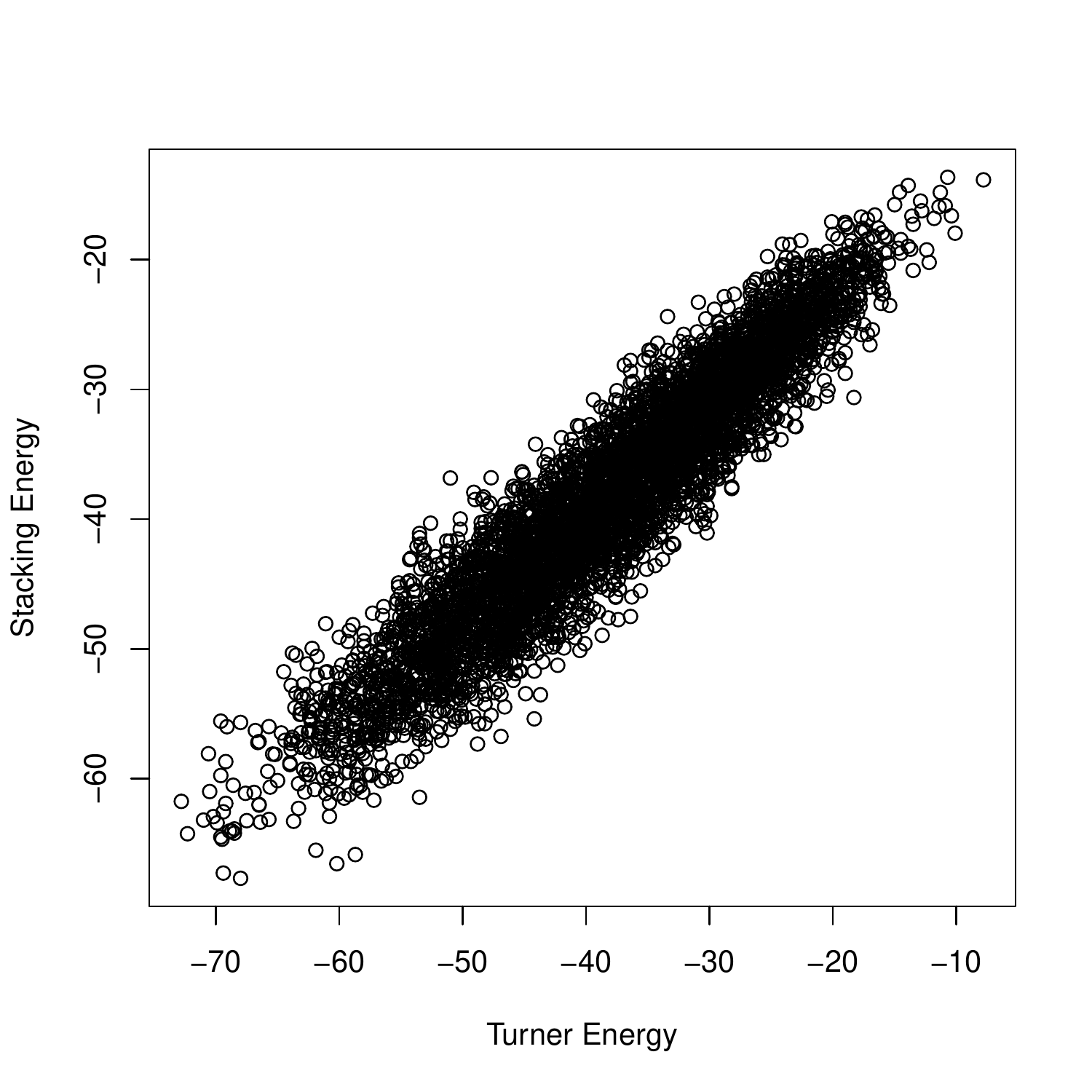}
  \caption{Validation of the trained parameters for the A) base pair
    model B) stacking model for predicting energies of the
    independently sampled sequences in the test set. We show
    correlation plots against the Turner energies reported by the
    ViennaRNA package.}
  \label{fig:training-cor}
\end{figure}

\begin{table}[b]
  \centering
  \caption{Trained weights for the base pair energy model.}
  \label{tab:basepairmodel}
  \begin{tabular}{c@{\quad}c@{\quad}c@{\quad}|@{\quad}c@{\quad}c@{\quad}c}
    \multicolumn{3}{c}{non-stacked} & \multicolumn{3}{c}{stacked}\\
    AU      & CG       & GU      & AU       & GC       & GU \\\hline
    1.26630 & -0.09070 & 0.78566 & -0.52309 & -2.10208 & -0.88474
  \end{tabular}
\end{table}


\begin{table}[b]
  \centering
  \caption{Trained weights for the stacking energy model (the rows specify $S_i,S_j$; the columns, $S_{i+1},S_{j-1}$; we do not show the symmetric weights).}
  \label{tab:stackingmodel}
  \begin{tabular}{c@{\quad}|@{\quad}c@{\quad}c@{\quad}c@{\quad}c@{\quad}c@{\quad}c@{\quad}c@{\quad}c}
       & AU & CG & GC & GU & UA & UG \\\hline
    AU &
-0.18826 &
-1.13291 &
-1.09787 &
-0.38606 &
-0.26510 &
-0.62086
    \\
    CG &
-1.11752 &
-2.23740 &
-1.89434 &
-1.22942 &
-1.10548 &
-1.44085
    \\
    GU &
-0.55066 &
-1.26209 &
-1.58478 &
-0.72185 &
-0.49625 &
-0.68876
    \\
  \end{tabular}
\end{table}

\section{Monotonicity of the partial derivatives within weight calibration}
\label{sec:weight_derivatives}
In weighted distributions, one witnesses a fairly predictable impact of the variation of any weight $\pi_\ell$ over the expected value $\mathbb{E}(E_\ell\mid \pmb{\pi})$, $\pmb{\pi}:= (\pi_1\cdots\pi_k)$, of the free energy $E_\ell$ of structure $\ell$.  Let us  first remind the probability of a sequence $S$ in the weighted distribution
\begin{align*}
  \mathcal{B}_{\pmb{\pi}}(S) &= \prod_{i=1}^{k} \pi_i^{-E_i(S)}, &
  \partfun{\pmb{\pi}}&=\sum_{S'}\mathcal{B}_{\pmb{\pi}}(S') &
    \text{and}& &
  \mathbb{P}(S\mid \pmb{\pi}) &= \frac{\mathcal{B}_{\pmb{\pi}}(S)}{\partfun{\pmb{\pi}}}.
  \end{align*}

Remind also the definition of the expectation of $E_\ell(S)$, for $S$ a Boltzmann-distributed sequence:
$$\mathbb{E}(E_\ell\mid \pmb{\pi}) = \sum_S E_\ell(S)\cdot \mathbb{P}(S\mid \pmb{\pi}).$$

We first remark that the partial derivative of $\mathcal{B}$ yields
\begin{align*}
  \frac{\partial \mathcal{B}_{\pmb{\pi}}(S)}{\partial \pi_\ell} = -E_\ell(S)\cdot \pi_\ell^{-E_\ell(S)-1}\prod_{\substack{i=1\\i\neq \ell}}^{k} \pi_i^{-E_i(S)} = \frac{-E_\ell(S)}{\pi_\ell}\cdot\mathcal{B}_{\pmb{\pi}}(S) = -E_\ell(S)\cdot \mathbb{P}(S\mid \pmb{\pi})\cdot \frac{\partfun{\pmb{\pi}}}{\pi_\ell }
\end{align*}
while the partial derivative of $\mathcal{Z}$ gives
\begin{align*}
  \frac{\partial \partfun{\pmb{\pi}}(S)}{\partial \pi_\ell} = \sum_{S'} -E_\ell(S')\cdot \mathbb{P}(S'\mid \pmb{\pi})\cdot \frac{\partfun{\pmb{\pi}}}{\pi_\ell }= -\mathbb{E}(E_\ell\mid \pmb{\pi})\cdot \frac{\partfun{\pmb{\pi}}}{\pi_\ell}.
\end{align*}
From these expressions, we conclude that
\begin{align*}
  \frac{\partial \mathbb{E}(E_\ell\mid \pmb{\pi})}{\partial \pi_\ell} &= \sum_S E_\ell(S)\cdot\frac{\partial \mathbb{P}(S\mid \pmb{\pi})}{\partial \pi_\ell}\\
  & = \sum_S E_\ell(S)\left(\frac{\frac{\partial \mathcal{B}_{\pmb{\pi}}(S)}{\partial \pi_\ell}}{\mathcal{Z}_{\pmb{\pi}}} - \frac{\frac{\partial \mathcal{Z}_{\pmb{\pi}}}{\partial \pi_\ell}\cdot\mathcal{B}_{\pmb{\pi}}(S)}{\mathcal{Z}_{\pmb{\pi}}^2}\right)\\
  & = \sum_S E_\ell(S)\left(\frac{-E_\ell(S)\cdot \mathbb{P}(S\mid \pmb{\pi})\cdot \frac{\mathcal{Z}_{\pmb{\pi}}}{\pi_\ell }}{\mathcal{Z}_{\pmb{\pi}}} + \frac{\mathbb{E}(E_\ell\mid \pmb{\pi})\cdot \frac{\mathcal{Z}_{\pmb{\pi}}}{\pi_\ell}\cdot\mathcal{B}_{\pmb{\pi}}(S)}{\mathcal{Z}_{\pmb{\pi}}^2}\right)\\
  & = \sum_S - \frac{E_\ell(S)^2\cdot\mathbb{P}(S\mid \pmb{\pi})}{\pi_\ell} + \frac{\mathbb{E}(E_\ell\mid \pmb{\pi})}{\pi_\ell}\sum_S E_\ell(S)\cdot \mathbb{P}(S\mid \pmb{\pi})\\
  & = -\frac{\mathbb{E}(E_\ell^2\mid \pmb{\pi}) - \mathbb{E}(E_\ell\mid \pmb{\pi})^2}{\pi_\ell}
\end{align*}
In the numerator of the above expression, one recognizes the variance of the distribution of $E_\ell$.
Remark that a variance is always non-negative and, in our case, also strictly positive for $\pmb{\pi}$, $\pi_\ell>0$, as soon as there exist at least two distinct sequences $S$ and $S'$ such that $E_\ell(S)\neq E_\ell(S')$.
Moreover, weights are always positive so the partial derivative in $\pi_\ell$ is always positive. Ultimately, it is always possible to increase (resp. decrease) the expected free energy of a structure by simply decreasing (resp. increasing) its weight, supporting our weight optimization procedure.
\newpage
\section{Detailed result of quality analysis}\label{sec:validity}
N/A values correspond to data that could not be obtained by the time of this submission for two main reasons: In the case of the initial sampling (left columns), they correspond to instances which, in conjunction with an expressive energy model, resulted in very high tree-width, leading to unreasonable memory requirements incompatible with our available computing resources; The case of missing values after optimization  (right columns), they indicate situations where the initial optimization took too long ($\approx$ 1 day) and was killed.
\subsection{RNAtabupath (2 structures)}
{
\begin{longtable}{@{}cr@{\hspace{1em}}r@{\hspace{1em}}r@{\hspace{1em}}r@{\hspace{1em}}r@{\hspace{2em}}r@{\hspace{1em}}r@{\hspace{1em}}r@{\hspace{1em}}r@{\hspace{1em}}r@{}}
&\multicolumn{2}{c}{Boltzmann}&\multicolumn{2}{l}{Uniform}&&\multicolumn{2}{l}{\begin{minipage}[b]{1.8cm} Boltzmann\\Optimized\end{minipage}}&\multicolumn{2}{l}{\begin{minipage}[b]{1.8cm} Uniform\\Optimized\end{minipage}}&\\ \toprule
Name&Mean&StDev&Mean&StDev&\%Gain&Mean&StDev&Mean&StDev&\%Gain\\ \midrule
alpha\_operon&17.18&4.18&25.86&5.97&51&2.12&0.81&2.36&0.88&11\\
amv&14.37&2.90&27.26&5.65&90&3.67&0.90&4.27&1.08&16\\
attenuator&10.74&3.04&21.44&4.77&100&1.39&0.55&1.81&0.80&30\\
dsrA&11.77&2.58&18.23&4.62&55&3.71&0.94&3.80&0.95& 2\\
hdv&22.08&4.88&35.48&6.68&61&3.63&1.16&4.57&1.48&26\\
hiv&41.22&6.74&73.00&9.40&77&16.87&3.45&29.89&5.36&77\\
ms2&13.35&3.73&19.79&4.55&48&2.41&0.76&2.66&0.97&10\\
rb1&17.70&4.35&37.38&7.14&111&2.64&0.90&3.82&1.31&45\\
rb2&16.63&4.16&25.17&5.85&51&2.75&0.86&3.07&0.97&12\\
rb3&22.47&4.71&41.18&7.09&83&3.79&1.01&5.05&1.52&33\\
rb4&26.15&4.30&47.89&6.80&83&10.57&&11.80&1.82&12\\
rb5&30.52&5.45&54.24&7.76&78&6.10&1.97&9.22&2.78&51\\
ribD&49.84&7.01&81.27&9.07&63&20.47&3.17&31.01&4.47&51\\
s15&11.00&2.97&18.37&4.53&67&1.96&0.64&2.12&0.77& 8\\
sbox&25.05&4.84&50.03&8.04&100&6.34&1.49&8.90&2.19&40\\
spliced&9.65&2.98&17.60&4.39&82&2.13&0.44&2.41&0.59&13\\
sv11&20.98&4.66&36.83&6.91&76&N/A&N/A&4.20&1.12&N/A\\
thim&29.35&5.39&48.32&6.95&65&8.69&1.90&12.12&2.54&39\\
\midrule
Average&21.67&4.38&37.74&6.45&74&5.84&1.31&7.95&1.76&28\\\bottomrule
\end{longtable}}
\subsection{RNAdesign (3 structures)}
{
\begin{longtable}{@{}cr@{\hspace{1em}}r@{\hspace{1em}}r@{\hspace{1em}}r@{\hspace{1em}}r@{\hspace{2em}}r@{\hspace{1em}}r@{\hspace{1em}}r@{\hspace{1em}}r@{\hspace{1em}}r@{}}
&\multicolumn{2}{c}{Boltzmann}&\multicolumn{2}{l}{Uniform}&&\multicolumn{2}{l}{\begin{minipage}[b]{1.8cm} Boltzmann\\Optimized\end{minipage}}&\multicolumn{2}{l}{\begin{minipage}[b]{1.8cm} Uniform\\Optimized\end{minipage}}&\\ \toprule
Name&Mean&StDev&Mean&StDev&\%Gain&Mean&StDev&Mean&StDev&\%Gain\\\midrule
sq100&20.67&4.90&31.78&5.46&54&5.81&1.53&8.33&1.92&43\\
sq10&19.91&5.24&26.62&5.59&34&3.50&0.81&3.88&0.94&11\\
sq11&19.38&4.81&27.76&5.69&43&2.81&0.77&3.31&0.89&18\\
sq12&17.19&3.62&22.62&4.98&32&2.70&0.65&2.79&0.68& 3\\
sq13&18.54&3.95&29.95&5.56&62&10.24&2.18&16.63&3.13&62\\
sq14&21.49&4.23&37.82&5.55&76&7.25&1.63&9.90&2.10&37\\
sq15&21.91&3.78&36.32&5.13&66&11.15&1.75&15.49&2.34&39\\
sq16&15.99&4.22&27.70&5.58&73&2.51&0.69&2.96&0.94&18\\
sq17&20.71&3.72&33.62&5.33&62&5.83&1.08&7.36&1.49&26\\
sq18&19.64&3.91&35.26&5.13&80&5.83&1.41&7.97&2.00&37\\
sq19&14.78&3.33&28.64&5.34&94&3.84&0.77&4.66&1.04&21\\
sq1&14.85&3.84&21.89&5.48&47&2.02&0.65&2.20&0.69& 9\\
sq20&17.39&3.82&27.90&4.93&60&3.93&1.04&5.18&1.48&32\\
sq21&21.64&3.94&36.27&5.16&68&8.28&1.82&11.41&2.40&38\\
sq22&17.65&4.13&31.57&5.45&79&5.66&1.51&10.02&2.50&77\\
sq23&19.44&3.66&32.38&5.40&67&12.05&1.85&18.69&3.18&55\\
sq24&18.01&4.01&27.92&5.30&55&3.59&0.93&4.25&1.10&19\\
sq25&17.54&5.20&26.65&6.11&52&1.35&0.34&1.45&0.40& 7\\
sq26&17.63&4.47&24.46&5.64&39&3.11&0.62&3.38&0.68& 9\\
sq27&17.93&3.61&28.51&5.30&59&7.55&1.49&10.96&2.38&45\\
sq28&14.78&3.71&29.35&5.99&99&2.32&0.49&2.66&0.62&15\\
sq29&17.61&3.72&26.95&5.08&53&4.44&0.98&5.22&1.21&18\\
sq2&22.43&4.19&37.25&5.15&66&10.34&1.71&14.32&2.35&39\\
sq30&19.11&4.04&31.75&5.54&66&6.06&1.61&8.99&2.19&48\\
sq31&19.97&3.87&31.92&5.20&60&5.12&1.19&6.34&1.55&24\\
sq32&21.96&4.00&31.65&5.23&44&5.24&1.29&6.84&1.76&31\\
sq33&15.11&4.30&22.75&5.52&51&1.78&0.43&2.08&0.63&17\\
sq34&25.00&4.17&37.46&5.03&50&12.46&1.84&18.73&2.57&50\\
sq35&19.25&4.12&32.72&5.35&70&5.90&1.25&7.96&1.82&35\\
sq36&12.77&2.91&26.76&5.39&110&3.19&0.77&3.61&0.94&13\\
sq37&18.71&3.88&32.44&5.62&73&3.24&0.90&4.21&1.25&30\\
sq38&14.55&3.56&29.02&5.29&99&2.79&0.86&3.70&1.22&33\\
sq39&22.75&3.95&33.06&4.98&45&9.20&1.76&11.87&2.33&29\\
sq3&20.38&4.07&36.86&5.60&81&4.69&1.19&6.30&1.72&34\\
sq40&19.46&4.35&31.76&5.41&63&4.05&0.94&4.96&1.19&23\\
sq41&14.41&3.46&27.85&5.57&93&3.33&0.71&3.91&0.88&17\\
sq42&18.67&3.87&35.20&5.55&89&5.32&1.24&6.96&1.56&31\\
sq43&17.32&3.97&31.63&5.38&83&3.53&0.89&4.27&1.15&21\\
sq44&19.37&3.84&31.75&5.09&64&5.52&1.29&7.26&1.60&31\\
sq45&18.53&4.77&27.40&5.53&48&2.54&0.67&2.84&0.74&12\\
sq46&17.92&3.80&28.83&5.09&61&4.45&0.90&5.41&1.20&22\\
sq47&15.09&3.78&28.23&5.82&87&3.02&0.78&3.45&0.85&14\\
sq48&18.90&4.23&36.27&5.78&92&4.69&1.08&6.71&1.67&43\\
sq49&20.13&4.69&27.24&5.39&35&3.15&0.76&3.66&0.95&16\\
sq4&19.31&4.69&29.25&5.34&51&2.36&0.79&3.24&1.15&37\\
sq50&17.01&3.92&28.43&5.35&67&3.15&0.73&3.54&0.86&12\\
sq51&22.18&4.14&36.07&5.20&63&11.56&1.79&16.57&2.88&43\\
sq52&19.36&3.93&31.04&5.20&60&11.04&2.21&16.25&3.04&47\\
sq53&13.78&3.72&24.48&5.46&78&1.65&0.39&1.88&0.53&14\\
sq54&18.07&3.59&33.99&5.57&88&8.15&1.57&12.35&2.59&51\\
sq55&21.93&4.68&30.11&5.21&37&4.59&0.82&5.43&1.13&18\\
sq56&17.83&3.90&31.69&5.46&78&4.22&0.75&4.84&0.93&15\\
sq57&16.30&3.97&28.48&5.45&75&1.96&0.58&2.28&0.71&16\\
sq58&12.74&2.86&32.97&5.80&159&3.77&0.82&4.89&1.22&29\\
sq59&20.44&4.54&33.02&5.57&62&6.05&1.50&8.34&2.08&38\\
sq5&9.97&2.72&17.94&4.93&80&1.29&0.33&1.41&0.39& 9\\
sq60&22.64&4.85&36.86&5.46&63&6.90&1.62&11.07&2.33&60\\
sq61&14.45&3.78&30.82&6.03&113&3.07&0.76&3.83&1.01&25\\
sq62&21.08&5.02&35.27&5.35&67&9.76&1.67&16.93&2.79&73\\
sq63&13.44&3.60&27.20&5.63&102&1.58&0.36&1.85&0.53&17\\
sq64&19.21&3.87&34.88&5.24&82&6.10&1.53&9.45&2.29&55\\
sq65&19.33&4.53&27.51&5.15&42&13.95&2.94&19.40&3.61&39\\
sq66&18.18&3.70&30.79&5.12&69&6.93&1.42&10.05&2.03&45\\
sq67&20.63&3.41&37.06&5.18&80&16.11&1.97&27.21&3.40&69\\
sq68&23.09&4.33&33.33&5.08&44&N/A&N/A&9.23&1.94&N/A\\
sq69&18.79&3.58&35.06&5.21&87&10.77&1.94&18.38&3.00&71\\
sq6&20.32&3.77&38.48&5.48&89&5.55&1.32&8.63&2.04&56\\
sq70&13.59&3.31&31.06&5.52&129&2.59&0.80&3.80&1.17&47\\
sq71&20.07&4.05&32.71&5.39&63&5.57&1.06&6.61&1.30&19\\
sq72&15.02&3.22&24.48&4.94&63&2.02&0.57&2.49&0.84&23\\
sq73&14.88&3.74&28.44&5.56&91&2.34&0.69&3.10&1.04&32\\
sq74&17.80&4.10&30.13&5.30&69&4.12&0.79&5.03&1.04&22\\
sq75&19.66&4.05&31.19&5.19&59&5.63&1.26&7.52&1.70&34\\
sq76&19.91&3.84&32.44&5.12&63&9.12&1.92&13.89&2.39&52\\
sq77&18.69&3.46&34.83&5.21&86&6.01&1.25&8.61&1.75&43\\
sq78&18.28&3.48&32.94&6.01&80&5.39&1.12&7.07&1.51&31\\
sq79&21.57&4.71&30.73&5.47&42&13.10&2.89&17.86&3.45&36\\
sq7&20.28&3.76&32.95&5.41&63&5.51&1.70&6.68&1.90&21\\
sq80&15.62&3.92&28.29&5.30&81&1.85&0.53&2.41&0.88&30\\
sq81&21.24&4.57&29.92&5.44&41&3.74&0.95&4.67&1.33&25\\
sq82&8.85&2.35&19.89&5.36&125&1.39&0.21&1.47&0.28& 6\\
sq83&14.13&3.60&28.92&5.45&105&2.77&0.76&3.52&1.01&27\\
sq84&15.60&3.98&24.12&5.26&55&2.54&0.47&2.85&0.65&12\\
sq85&15.56&3.39&26.09&5.18&68&2.19&0.62&2.70&0.82&23\\
sq86&21.54&4.95&34.21&5.77&59&4.66&1.21&5.92&1.51&27\\
sq87&16.45&3.85&32.65&5.81&98&3.50&0.62&4.14&0.86&18\\
sq88&20.05&5.09&33.45&5.86&67&2.50&0.72&3.17&0.99&27\\
sq89&14.77&4.20&26.15&5.59&77&1.79&0.50&2.01&0.62&13\\
sq8&15.85&4.14&30.34&6.02&91&2.85&0.79&3.85&1.15&35\\
sq90&15.25&3.99&29.83&5.80&96&2.10&0.64&2.90&1.06&38\\
sq91&17.29&3.93&27.19&5.08&57&3.69&1.11&4.48&1.31&21\\
sq92&19.06&3.57&33.07&5.34&74&4.77&1.11&6.83&1.78&43\\
sq93&19.48&4.11&29.17&5.65&50&3.29&0.80&3.86&1.03&17\\
sq94&14.72&3.21&27.05&5.43&84&3.98&0.77&4.66&0.98&17\\
sq95&17.60&3.72&33.42&5.38&90&12.86&2.20&22.76&3.84&77\\
sq96&15.35&3.84&31.62&5.68&106&2.69&0.67&3.52&1.02&31\\
sq97&18.28&3.89&29.22&5.03&60&4.31&1.13&5.32&1.42&23\\
sq98&18.84&4.57&29.30&5.42&56&2.78&0.73&3.59&0.97&29\\
sq99&19.74&4.19&30.95&5.28&57&3.94&0.98&4.89&1.33&24\\
sq9&17.30&4.26&26.15&5.72&51&1.66&0.44&1.79&0.55&8\\
\midrule
Average&21.67&4.38&37.74&6.45&74&5.84&1.31&7.95&1.76&28\\\bottomrule
\end{longtable}}
\subsection{RNAdesign (4 structures)}
{
\begin{longtable}{@{}cr@{\hspace{1em}}r@{\hspace{1em}}r@{\hspace{1em}}r@{\hspace{1em}}r@{\hspace{2em}}r@{\hspace{1em}}r@{\hspace{1em}}r@{\hspace{1em}}r@{\hspace{1em}}r@{}}
&\multicolumn{2}{c}{Boltzmann}&\multicolumn{2}{l}{Uniform}&&\multicolumn{2}{l}{\begin{minipage}[b]{1.8cm} Boltzmann\\Optimized\end{minipage}}&\multicolumn{2}{l}{\begin{minipage}[b]{1.8cm} Uniform\\Optimized\end{minipage}}&\\ \toprule
Name&Mean&StDev&Mean&StDev&\%Gain&Mean&StDev&Mean&StDev&\%Gain\\\midrule
sq100&21.70&4.41&33.19&4.93&53&12.52&2.62&17.14&3.10&37\\
sq10&20.66&5.16&27.00&5.24&31&N/A&N/A&8.02&2.05&N/A\\
sq11&22.23&4.23&30.73&5.31&38&6.09&1.21&7.70&1.70&26\\
sq12&18.30&3.55&23.38&4.93&28&3.75&0.53&3.94&0.64& 5\\
sq13&21.01&3.57&31.66&4.93&51&15.63&1.69&21.61&2.69&38\\
sq14&22.44&3.16&39.77&5.28&77&17.56&2.21&25.04&3.21&43\\
sq15&22.66&3.72&36.31&5.01&60&17.45&2.38&26.80&3.34&54\\
sq16&18.22&4.12&28.54&5.33&57&3.27&0.74&3.78&0.90&16\\
sq17&23.04&4.01&35.87&5.28&56&12.76&2.54&18.43&2.86&44\\
sq18&21.20&3.49&35.78&5.39&69&8.15&1.47&10.75&2.02&32\\
sq19&15.36&2.73&32.12&4.95&109&13.22&1.73&23.64&3.41&79\\
sq1&15.14&3.83&22.32&5.28&47&2.55&0.56&2.66&0.56& 4\\
sq20&16.75&3.09&29.25&4.74&75&9.26&1.67&13.40&2.41&45\\
sq21&23.98&3.72&38.18&5.14&59&16.15&2.10&23.31&3.16&44\\
sq22&20.38&3.81&32.90&5.12&61&8.42&1.44&12.39&2.03&47\\
sq23&18.39&3.36&32.82&5.34&78&12.07&1.93&19.58&3.26&62\\
sq24&20.05&4.10&29.72&5.29&48&5.40&1.15&7.28&1.68&35\\
sq25&21.48&5.35&28.83&5.78&34&2.95&0.66&3.30&0.77&12\\
sq26&20.83&4.63&27.39&5.25&32&4.20&0.76&4.80&0.96&14\\
sq27&17.14&3.48&29.41&5.24&72&7.58&1.45&11.67&2.34&54\\
sq28&23.50&4.48&33.84&6.07&44&5.59&1.03&6.69&1.39&20\\
sq29&18.14&3.65&27.22&4.85&50&5.54&0.92&6.26&1.17&13\\
sq2&N/A&N/A&N/A&N/A&N/A&N/A&N/A&N/A&N/A&N/A\\
sq30&19.40&3.76&32.40&4.95&67&8.89&1.56&12.39&2.19&39\\
sq31&24.97&3.80&32.68&5.19&31&8.77&1.18&13.05&2.15&49\\
sq32&19.62&3.76&31.03&4.97&58&5.50&1.29&7.57&1.71&38\\
sq33&19.51&4.25&27.65&5.12&42&4.23&0.91&5.37&1.16&27\\
sq34&25.60&3.50&37.58&4.90&47&13.41&1.77&20.62&2.68&54\\
sq35&21.87&3.93&35.50&5.73&62&9.66&1.56&13.80&2.31&43\\
sq36&11.65&2.59&26.29&5.48&126&3.33&0.60&3.86&0.90&16\\
sq37&25.80&4.48&36.19&4.90&40&15.84&2.03&23.46&2.85&48\\
sq38&17.74&3.43&32.11&5.14&81&4.95&0.93&6.20&1.38&25\\
sq39&21.54&3.83&32.74&5.21&52&9.79&1.81&13.10&2.29&34\\
sq3&20.19&4.06&39.87&5.23&97&14.60&2.26&26.21&3.32&79\\
sq40&18.68&4.18&33.41&5.62&79&4.32&0.92&5.49&1.30&27\\
sq41&15.86&3.46&30.06&5.30&89&5.56&0.94&6.59&1.17&19\\
sq42&21.79&3.25&38.72&5.11&78&N/A&N/A&23.27&3.45&N/A\\
sq43&23.21&3.99&33.89&5.33&46&7.27&1.30&9.08&1.76&25\\
sq44&21.65&3.59&33.47&4.99&55&8.04&1.42&10.91&1.97&36\\
sq45&19.01&4.79&27.79&5.58&46&3.45&0.65&3.81&0.83&10\\
sq46&20.08&3.60&30.69&4.89&53&7.68&1.30&9.96&1.69&30\\
sq47&17.29&3.73&29.94&5.26&73&4.52&1.03&5.24&1.20&16\\
sq48&28.78&4.03&38.24&5.57&33&16.87&1.76&25.52&3.19&51\\
sq49&19.19&4.75&27.33&5.42&42&3.14&0.73&3.66&0.86&17\\
sq4&19.71&4.46&29.44&5.30&49&3.35&0.87&4.25&1.18&27\\
sq50&19.36&4.11&30.60&5.49&58&4.99&0.79&5.71&0.98&14\\
sq51&N/A&N/A&38.68&5.02&N/A&N/A&N/A&27.99&2.48&N/A\\
sq52&21.40&3.37&33.80&5.43&58&16.30&2.43&25.10&3.67&54\\
sq53&15.81&3.50&28.75&5.32&82&4.29&1.05&6.22&1.53&45\\
sq54&18.16&3.45&35.10&4.86&93&15.32&2.23&26.27&3.37&71\\
sq55&22.47&4.43&31.94&5.15&42&11.05&2.41&15.19&2.98&37\\
sq56&21.11&4.12&35.25&5.93&67&12.27&2.00&18.68&2.97&52\\
sq57&18.94&3.91&30.86&5.51&63&3.69&0.71&4.51&0.90&22\\
sq58&19.68&2.97&35.80&5.41&82&9.78&1.44&16.76&2.66&71\\
sq59&20.35&3.75&33.08&5.20&63&17.24&2.82&24.28&3.93&41\\
sq5&11.60&2.54&19.61&5.12&69&1.92&0.35&2.12&0.49&11\\
sq60&21.97&4.37&38.39&5.44&75&17.44&2.44&29.57&3.89&70\\
sq61&19.07&3.46&34.14&5.24&79&6.90&1.33&10.30&1.97&49\\
sq62&25.71&5.09&34.50&4.88&34&12.33&1.94&17.33&2.62&41\\
sq63&17.21&3.45&30.71&5.37&78&5.49&1.36&7.38&1.84&34\\
sq64&19.89&3.52&35.51&4.96&79&14.00&2.12&22.04&3.12&57\\
sq65&18.69&4.06&29.18&5.34&56&12.55&2.26&17.63&2.94&41\\
sq66&17.50&3.40&30.59&5.16&75&8.71&1.62&13.67&2.43&57\\
sq67&22.63&4.21&36.73&5.06&62&17.17&2.21&26.90&3.34&57\\
sq68&N/A&N/A&35.00&5.13&N/A&N/A&N/A&17.96&2.39&N/A\\
sq69&N/A&N/A&36.07&4.91&N/A&N/A&N/A&24.53&2.08&N/A\\
sq6&22.32&3.77&40.27&5.41&80&6.25&1.39&9.46&2.08&51\\
sq70&16.97&4.07&31.84&5.51&88&3.64&0.88&5.09&1.32&40\\
sq71&19.34&4.07&31.81&5.24&64&5.09&0.95&6.26&1.20&23\\
sq72&16.60&3.01&29.22&4.92&76&6.20&1.18&7.67&1.57&24\\
sq73&18.64&3.44&30.86&5.16&66&12.54&2.35&21.15&4.05&69\\
sq74&17.24&3.88&31.19&5.47&81&13.12&3.31&21.49&4.11&64\\
sq75&18.44&3.60&31.25&5.19&70&7.24&1.37&11.19&2.21&55\\
sq76&20.76&3.38&33.74&4.94&63&16.90&2.38&25.18&3.38&49\\
sq77&21.92&3.29&35.83&5.26&63&10.59&1.63&13.96&2.25&32\\
sq78&23.60&3.17&36.23&5.47&54&19.14&2.52&26.07&3.61&36\\
sq79&24.77&4.43&32.49&5.23&31&18.07&2.89&23.78&3.64&32\\
sq7&19.94&3.54&32.78&5.29&64&5.37&1.65&6.59&1.89&23\\
sq80&18.63&3.88&29.75&5.27&60&2.87&0.64&3.64&1.01&26\\
sq81&23.57&3.63&36.02&5.24&53&17.71&2.20&25.76&3.17&45\\
sq82&13.21&3.25&25.45&5.00&93&3.05&0.67&4.06&1.00&33\\
sq83&14.37&3.21&31.96&5.41&122&3.86&0.84&5.55&1.31&44\\
sq84&19.05&3.89&29.35&5.00&54&5.21&0.99&6.38&1.32&23\\
sq85&17.32&3.91&30.80&5.43&78&10.61&2.09&17.53&2.93&65\\
sq86&25.19&4.04&36.79&5.66&46&9.04&1.93&14.14&2.62&56\\
sq87&16.91&3.89&32.84&5.34&94&4.94&0.79&6.08&1.18&23\\
sq88&22.92&4.99&35.87&5.65&56&5.22&1.00&6.58&1.30&26\\
sq89&17.52&3.91&29.24&5.45&67&4.14&0.86&4.83&1.11&17\\
sq8&25.56&4.49&33.42&5.21&31&N/A&N/A&12.55&2.05&N/A\\
sq90&18.67&3.62&33.21&5.33&78&5.40&0.98&7.26&1.62&34\\
sq91&18.00&3.83&26.82&4.85&49&8.01&1.72&11.08&2.41&38\\
sq92&20.51&3.28&37.62&4.84&N/A&17.62&2.17&29.08&3.47&65\\
sq93&22.63&4.66&33.40&5.25&48&11.43&2.16&16.37&3.00&43\\
sq94&14.71&3.25&28.06&5.39&91&3.81&0.72&4.51&0.97&18\\
sq95&N/A&N/A&36.93&5.11&N/A&N/A&N/A&23.98&2.67&N/A\\
sq96&17.20&3.70&33.01&5.46&92&5.98&1.20&8.09&1.71&35\\
sq97&19.78&3.85&30.01&5.09&52&6.05&1.25&7.45&1.60&23\\
sq98&21.39&4.33&31.13&5.16&46&4.77&0.77&5.67&1.03&19\\
sq99&20.33&4.39&29.97&5.07&47&4.05&1.07&4.95&1.28&22\\
sq9&18.86&4.22&29.51&5.22&56&3.93&0.88&5.06&1.35&29\\
\\\midrule
Average&19.94&3.84&32.29&5.24&63&8.77&1.48&13.13&2.13&37\\ \bottomrule
\end{longtable}}


\end{document}